\documentclass[reprint,aps,pra,onecolumn,superscriptaddress,longbibliography,amsmath,amssymb]{revtex4-2}

\usepackage{bm}
\usepackage{mathtools}
\usepackage{booktabs}
\usepackage{xcolor}
\usepackage{hyperref}
\hypersetup{colorlinks=true,linkcolor=blue,citecolor=blue,urlcolor=blue}
\providecommand{\Eprint}[2]{#2}

\newcommand{\Tr}{\operatorname{Tr}}

\newcommand{\id}{\mathbb I}

\newcommand{\TV}{\operatorname{TV}}

\newcommand{\CPT}{\mathsf{CPT}}
\newcommand{\Free}{\mathsf F}
\newcommand{\Proc}{\mathsf P}
\newcommand{\Inst}{\mathfrak I}
\newcommand{\Cost}{\mathsf c}
\newcommand{\eps}{\varepsilon}

\newtheorem{theorem}{Theorem}
\newtheorem{proposition}{Proposition}
\newtheorem{lemma}{Lemma}
\newtheorem{corollary}{Corollary}
\newtheorem{definition}{Definition}
\newtheorem{remark}{Remark}
\newenvironment{proof}{\par\noindent\textit{Proof.}\ }{\hfill$\Box$\par}

\begin{document}

\title{Temporal-Access Resource Theory for Continuous Process Tensors}

\author{He Wang}
\email{dr\_hewang19@163.com}

\begin{abstract}
The uncontrolled endpoint channel can conceal how well a quantum process can protect a logical task under time-resolved control. We formulate a fixed-task resource theory of continuous process tensors with physical access interfaces, entrance-erasing free objects, and budgeted free conversions. Complete-record total-variation error measures performance; the calibrated capability is monotone under conversions that preserve a common free baseline. For unknown-state transmission the baseline is exactly $1-d^{-2}$. Two qubit processes with the same endpoint channel have optimal errors $0$ and $p=\sin^2(2gh)$ at the same finite pulse budget in an ideal-transparent system-only window architecture, forcing any endpoint-only predictor to err by at least $p/2$ on one of them. An exactly solvable family shows how the allocation of a fixed interaction phase around a control window sets the attainable error, and a nontrivial perfect echo requires and attains total control action $\pi$ in that architecture. The separation extends directly to every qubit unitary target at a budget that includes its boundary gate, while unitary tasks in any dimension have baseline $1-d^{-2}$ whenever a budget-feasible free realization gives the depolarizing output. Complete-record implementation bounds quantify coherent and dissipative control imperfections. Continuous-process randomized compilation supplies physical instrument pullbacks; compatible completion conditions extend their action, and certified dynamical bounds yield achievable task-error and capability guarantees. The resulting theory identifies temporal structure and control investment that endpoint characterization alone cannot resolve.
\end{abstract}

\maketitle

\section{Introduction}

An endpoint channel records the uncontrolled input--output transformation of a quantum device. It does not record when an interaction occurred or whether a finite control window can reverse it. Those missing data determine the best error achievable for a prescribed logical task. Process tensors and quantum combs retain the response to internal interventions \cite{pollock2018,milz2021,chiribella2009}; operational distinguishability compares the resulting experiments \cite{zambon2024}. Continuous process descriptions also accommodate finite-duration instruments and controls \cite{wassner2026}, while dilation results connect such descriptions to system--environment realizations \cite{fuksa2026}. The operational question is how much internal time access improves a fixed task under a fixed physical budget.

Resource theories organize this question through processes, free transformations, and monotones \cite{chitambar2019}. For multi-time processes, the admissible superprocess and its interface determine which temporal correlations can be used \cite{berk2021}. Non-Markovian correlations can be consumed to reduce noise \cite{berk2023}; subsequent work clarified the dependence of dynamical monotones and noise reducibility on restricted control \cite{berk2025monotones,berk2025quantifiers}. Here the primary quantity is instead the optimized error relative to a preassigned logical process. An entrance-erasing process supplies the free benchmark, and a physical conversion is charged to the available budget.

We define that error by the total variation distance of complete experimental records, using the same task tests for the actual and ideal processes. The theory fixes the logical task, the temporal input interface, and the allowed control class before optimization. It constructs a nonempty free family by erasing the complete protected entrance wire and excluding input-bearing bypasses. Causal experiment pullback gives composition closure and a budgeted performance ordering; a common free benchmark turns the calibrated improvement into a monotone. For reference-assisted unknown-state transmission the free benchmark is $1-d^{-2}$.

Our main physical result is an exact, finite-budget separation. Two realizable qubit processes have the same uncontrolled endpoint channel, yet their best errors under the same ideal-transparent system-only window controls are $0$ and $p=\sin^2(2gh)$. Exact endpoint tomography therefore leaves a worst-case capability prediction error of at least $p/2$. A fixed-total-phase family identifies the useful time structure: placing interaction on both sides of a control window permits cancellation, whereas placing it entirely before the window does not. The balanced echo has a sharp total action threshold of $\pi$, attained by bounded-duration pulses. Common commuting window coupling and certified control imperfections admit quantitative stability bounds.

The exact qubit result extends physically to every target $U\in U(2)$ when its implementation cost is included. Separately, output-task and joint input--output covariance establish precise conditional transport statements, and same-dimensional unitary targets retain the free benchmark $1-d^{-2}$ at budgets admitting a depolarizing free realization. Continuous-process randomized compilation (CPRC) \cite{wang2026cprc} supplies a continuous family of physical instrument conversions under its covariance and realization conditions. We establish the required experiment-pullback structure and use the complete-record dynamical guarantee of Ref.~\cite{wang2026cprc} to certify achievable error at a declared budget. The resulting distinction is sharp: the finite-window optima are proved within their stated control class, and the CPRC estimates certify selected realizable protocols.

\section{Operational setting and task interface}
\label{sec:setting}

\begin{table}[t]
\caption{Principal symbols. Distinct fonts distinguish access classes,
physical registers, and operational quantities.}
\label{tab:notation}
\begin{ruledtabular}
\normalsize
\begin{tabular}{ll}
Symbol & Meaning\\
$\Lambda,\Xi_\tau$ & Actual process and prescribed ideal task\\
$a,\tau,B$ & Input access structure, output task, control budget\\
$\mathcal B_{\tau,a}^{\rm feas}$ & Budgets admitting at least one task control\\
$\Inst_a,\Inst_\tau$ & Accessible experiments and fixed task tests\\
$d$ & Protected system dimension; $d_a$ denotes the operational distance at access structure $a$\\
$\eps_\tau,E_{\tau,a,B},R_{\tau,a,B}$ & Task error, optimized error, calibrated advantage; unsubscripted $R$ is the tester reference\\
$\Free_a,\mathsf G$ & Free processes and preparation generators\\
$Q,R,M$ & Protected input, tester reference, complete entrance wire\\
$G_{f\to c},\mathfrak G_{\rm end}$ & Interface restriction and endpoint restriction\\
$G_{\rm app}(t),\mathcal A_F(G_{\rm app})$ & Apparatus field and its compiled physical field; $G_{f\to c}$ is the coarse restriction\\
$F,K_F$ & System frame and its Hamiltonian generator\\
$A_q,\Cost(P)$ & Generator norm cost and full implementation cost\\
$Z_t,\pi_t,P_{s,t}$ & Control state, marginal law, transition kernel\\
$N_{\rm fr},N$ & Frame alphabet size and holding-interval count
\end{tabular}
\end{ruledtabular}
\end{table}

Throughout, $\|\cdot\|_\diamond$ is the full diamond norm; the channel or
generator discrepancy called a half-diamond distance is
$\tfrac12\|\cdot\|_\diamond$.  The factor $1/2$ is included explicitly in
every error radius below.

\subsection{Continuous processes and access structures}

Let $S$ be a finite-dimensional system with a distinguished protected input $Q$ at time zero, and let $[0,T]$ be the physical time interval.  A reference $R$ may initially be entangled with $Q$; it is held by the task tester and is inaccessible to the controller until the terminal measurement.  Controller ancillas and random seeds are initially independent of $QR$.  Access to later logical ports is declared explicitly and cannot be used to reinsert a stored copy of the protected input.

An allowed continuous instrument $\mathcal J$ has a standard Borel record space $(X_{\mathcal J},\Sigma_{\mathcal J})$.  Its record up to $t$ generates a filtration, and its feedback maps are adapted to that filtration.  An instrument is countably additive on record events, completely positive under finite ancillary extension, and normalized when outcomes are summed.  Its Hamiltonian or channel realization and its temporal supports belong to the declared hardware class.  Supports may be finite unions of half-open intervals, with a separate terminal readout.  Continuous records are compared on this same space, or through a fixed, process-independent relabeling specified before optimization.  The experiment class includes causal prefix restrictions, deterministic completion, classical outcome coarse-graining, and the declared ancillary extensions.  Every admissible experiment is a finite physical network or an admissible continuous limit of such networks.  Its completion topology is fixed as part of the instrument class.

The set of complete experiments is denoted by $\Inst_a$, where
\begin{equation}
 a=(T,\delta t,\mathcal A,\mathcal K,\mathcal R)
 \label{eq:access-structure}
\end{equation}
records the time resolution $\delta t$, the allowed supports and physical instrument class $\mathcal A$, the ancillary-system class $\mathcal K$, and the record/feedback class $\mathcal R$, together with their completion topology.  We call $a$ an \emph{access structure}.  Peak amplitude and action are protocol budgets introduced below; a resolution label by itself does not grant arbitrary pulses.  Write $a_f\succeq a_c$ when a physical, outcome-compatible embedding $\iota_{c\to f}:\Inst_{a_c}\to\Inst_{a_f}$ preserves ancillary extensions, deterministic completion, causal prefixes, and the chosen continuous limits.  This access order concerns input experiments, not an enlargement of the output task tests.

A cPT $\Lambda\in\CPT_a$ is, at the operational level used here, a continuous Born functional on the span of admissible instrument events.  It is multilinear in distinct operation slots and affine under classical instrument randomization.  Positivity holds under every declared finite ancillary extension; deterministic experiments give unit mass; and a prefix marginal depends only on the corresponding prefix experiment.  Its event probabilities are countably additive.  These requirements define an \emph{operational cPT}; they do not by themselves assert an extension to every mathematically possible quantum operation.  We write $\CPT_a^{\rm real}\subseteq\CPT_a$ for the subclass supplied by a declared finite comb, CP-instrument network, Hamiltonian dilation, or an admissible continuum construction with the same properties.  All physical control and CPRC claims below apply to this realizable subclass, or conditionally to an operational class when the stated pullback extension exists.  No positivity condition is imposed on individual response amplitudes.  For $\mathcal J\in\Inst_a$, its complete output probability measure is written
\begin{equation}
 p_\Lambda^\mathcal J(dx)=\Lambda[\mathcal J](dx).
 \label{eq:born-cpt}
\end{equation}
For a finite set of slots this reduces to the process-tensor Born rule \cite{chiribella2009,pollock2018}.  A continuum realization may be supplied by the operational completion and dilation constructions of Refs.~\cite{wassner2026,fuksa2026}, with their hypotheses and topologies retained.  Our conversion results require the stated instrument structure and continuity, and apply to any realization satisfying them.  In a regular Fock representation one may write schematically
\begin{equation}
 |\Lambda_T\rangle=\sum_{n\ge0}\sum_{\boldsymbol\nu}
 \int_{0<t_1<\cdots<t_n<T}
 C_{\boldsymbol\nu}^{(n)}(\boldsymbol t)
 |\boldsymbol\nu;\boldsymbol t\rangle\,d\boldsymbol t,
 \label{eq:cpt-fock}
\end{equation}
where the coefficients are complex response amplitudes rather than probabilities.  Equation~\eqref{eq:cpt-fock} is a representation when its sector expansion converges in the chosen continuum space; it is not the definition of $\CPT_a$.  None of the resource monotonicity or separation proofs uses an unweighted Fock norm.  Operationally indistinguishable functionals may be identified; on a separating physical process space no additional quotient is needed.

The operational distance associated with $a$ is
\begin{equation}
 d_a(\Lambda,\Lambda'):=\sup_{\mathcal J\in\Inst_a}
 \TV\!\left(p_\Lambda^\mathcal J,p_{\Lambda'}^\mathcal J\right),
 \qquad
 \TV(P,Q):=\frac12\int|dP-dQ|.
 \label{eq:operational-distance}
\end{equation}
It is a pseudometric when $\Inst_a$ is not separating.  We identify processes with identical event laws on $\Inst_a$ and take operational closures in this quotient; $d_a$ is a metric there.  A closure statement in this metric concerns accessible event laws, without specifying a unique microscopic dilation.  The completion topology of instruments is a separate part of $a$: continuity of a Born functional in that topology does not assert density of regular dilations in $d_a$.  For every event $A$ in the complete record,
\begin{equation}
 |p_\Lambda^\mathcal J(A)-p_{\Lambda'}^\mathcal J(A)|
 \le d_a(\Lambda,\Lambda').
 \label{eq:event-bound}
\end{equation}

\subsection{Logical tasks}

A concrete instrument completion is supplied by finite-dimensional, finite-slot combs with finite outcome sets and a fixed finite ancillary dimension.  Instruments are tuples of CP maps equipped with the sum of their diamond-norm distances.  Outcome summation, deterministic completion, prefix restriction, and declared ancillary extension preserve the comb normalization equations.  Each Born probability is a finite multilinear contraction of Choi operators, hence continuous in this topology.  Refining the slots gives compatible families; a continuous limit is admitted when its Born functional is continuous in the specified instrument completion and retains these normalization and prefix identities.  This example establishes a nonempty operational model without identifying slot refinement with convergence in the uniform process metric $d_a$.

A logical task is a pair $\tau=(\Xi_\tau,\Inst_\tau)$, where $\Xi_\tau\in\CPT_{a_\tau}$ is fixed before any control optimization and $\Inst_\tau\subseteq\Inst_{a_\tau}$ is the set of logical experiments whose output distribution is to be protected.  Whenever a process at an access structure $a$ is compared with this task, the declared task family is required to satisfy
\[
 \Inst_\tau\subseteq\Inst_a\cap\Inst_{a_\tau};
\]
thus the same experiment is physically available at both interfaces.  Claims about a physically implemented task additionally require $\Xi_\tau\in\CPT_{a_\tau}^{\rm real}$.  The same logical experiment must be used when comparing an actual process with $\Xi_\tau$; separate optimization of the two sides would turn an error measure into a discrimination protocol.

The complete record space is part of the task interface, rather than an implementation detail.  By default it contains all declared logical instrument outcomes and the terminal readout record; internal control paths, random-frame labels, and sampling variables are averaged out.  If a control label is retained, its measurable label space and probability law must be fixed before optimization and the ideal task must carry the same label law.  Actual and ideal processes are therefore compared on one common complete record space, with any further coarse-graining specified as a fixed measurable record map.  This convention applies to both endpoint and multi-time tasks.

For $\Lambda$ expressed at the task interface, define
\begin{equation}
\eps_\tau(\Lambda):=\sup_{\mathcal J\in\Inst_\tau}
 \TV\!\left(p_\Lambda^{\mathcal J},p_{\Xi_\tau}^{\mathcal J}\right).
 \label{eq:task-error}
\end{equation}
The ideal task is the fixed resource-rich reference.  Environmental memory does not by itself determine task capability: a process with memory can have zero task error, and a memoryless channel can outperform the free baseline.  Resource advantage will mean an optimized error below that baseline, not simply a nonzero error.  The task family, including allowed input states and terminal measurements, is fixed across all access and control comparisons.  In particular, $0\le\eps_\tau\le1$ and $\eps_\tau(\Lambda)\le d_{a_\tau}(\Lambda,\Xi_\tau)$.

\subsection{Declared process and implementation error}

A cPT contains the environmental action assigned to the actual process, including its dissipation and memory. A control or instrument implementation error compares a realized conversion with its declared reference implementation on that process. If an additional dissipative mechanism belongs to the system--environment dynamics being studied, it defines a different actual cPT whose task error must be recomputed. The complete-record discrepancy used below keeps this distinction explicit and applies to any fixed task interface.

\section{Resource theory for fixed tasks}
\label{sec:resource}

Fixing a task and imposing ideal transparency are distinct choices.  We choose controls whose action on the compatible ideal input process is the prescribed ideal output process.  For example, terminal $Z$ corrects an actual channel $\operatorname{Ad}_Z$ to the identity task, although it sends the ideal identity channel to $\operatorname{Ad}_Z$.  Such a correction can be assessed using Eq.~\eqref{eq:task-error}, but is outside the ideal-transparent architecture studied here.  This restriction fixes the conversion theory and the control-relative optima below.

\subsection{Replacement seeds and the free set}

We use an explicit input-erasing architecture.  Each interface has the same protected external input $Q$ of dimension $d$, a designated process entrance, and a compatible ideal reference $\Xi_a$.  At the entrance a free seed erases the entire incoming wire and prepares a state independent of the unknown input.  On an unencoded $Q$ wire, a one-boundary seed is
\begin{equation}
 \mathcal R_\sigma(\rho)=\sigma\Tr\rho,
 \qquad \sigma\succeq0,\quad\Tr\sigma=1,
 \label{eq:replacement-channel}
\end{equation}
If an allowed encoding enlarges the incoming wire to $M$, the corresponding seed is $\mathcal R^{M}_{\sigma_M}$ on all of $M$, including every encoded ancillary component.  In particular,
\begin{equation}
 (\mathcal R^{M}_{\sigma_M}\otimes\mathrm{id}_R)(\chi_{MR})
 =\sigma_M\otimes\Tr_M\chi_{MR}.
 \label{eq:complete-entrance-erasure}
\end{equation}
The symbol $M$ is a typed entrance system, not an arbitrary auxiliary extension of the replacement channel on $Q$.  An interface declaration therefore includes an admissible encoding type $Q\otimes A_{\rm in}\to M\otimes D$ and the corresponding replacement seed on the whole $M$.  A reference or tester ancilla is not part of $M$ and is never erased by the controller.  If a hardware model permits several entrance types, each type has its own seed family and ideal reference; a conversion may change type only through a declared typed encoder and a seed that erases its complete output.  Thus the equality above is not inferred from $\mathcal R_\sigma^Q\otimes\mathrm{id}_{A_{\rm in}}$.
The continuous seed follows this erasure by an admissible prepared evolution, possibly with a memory-bearing environment and downstream records.  Controller memory, the prepared environment, and downstream preparation parameters are initially independent of $QR$.  The reference is untouched.  Consequently, for a fixed downstream strategy $\mathcal L$ preceding terminal readout,
\begin{equation}
 (\Lambda_{\mathcal L}\otimes\mathrm{id}_R)(\rho_{QR})
 =\omega_{\mathcal L}\otimes\rho_R,
 \qquad \rho_R=\Tr_Q\rho_{QR},
 \label{eq:input-erasing}
\end{equation}
where $\omega_{\mathcal L}$ includes every accessible controller record.  The equation concerns the output state before measurement.  Joint measurement probabilities can depend on $\rho_R$; declaring all probabilities input independent would incorrectly exclude reference-assisted testers.  Downstream interventions may supply new prescribed states, but cannot bring the protected input or its reference back into the controller.

Let $\mathsf G$ consist of input-independent prepared dynamics after erasure, independent ancillary preparation, ancillary discarding, classical mixing and record maps, and physical interface restrictions.  Here $\operatorname{Gen}_a$ means all finite networks built from entrance-erasing seeds and these generators whose output interface is $a$.  Interface-changing generators have specified source and target types; they are not self-maps of a single $\CPT_a$.  Define
\begin{equation}
 \Free_a:=\operatorname{cl}_{d_a}\!\left[
 \operatorname{Gen}_a\left(\{\mathcal R_\sigma\},\mathsf G\right)
 \right].
 \label{eq:free-set}
\end{equation}
Each finite generator network obeys Eq.~\eqref{eq:input-erasing}.  Its operational limit remains indistinguishable from an entrance-erasing process in every allowed experiment: event-law equality passes to a $d_a$ limit, and any allowed test witnessing input dependence would contradict that convergence.  For a complete-record interface, this closure is understood in the declared record topology.  If $\Phi_n$ are finite input-erasing maps, $\Phi_n(\rho)=\omega_n\operatorname{Tr}\rho$, and $\Phi_n\to\Phi$, then fixing one normalized input $\rho_0$ gives $\omega_n=\Phi_n(\rho_0)\to\omega:=\Phi(\rho_0)$.  Continuity of the ancillary extension and tensoring with a reference therefore preserves
\begin{equation}
(\Phi\otimes\operatorname{id}_R)(\rho_{QR})=\omega\otimes\rho_R .
\label{eq:input-erasing-limit}
\end{equation}
For continuous records, $\omega_n\to\omega$ is read as convergence of the corresponding record-valued measures (or, on a restricted interface, of all declared event laws); the same factorization is then asserted at that operational level.  When the output is finite-dimensional and state-tomographic, this reduces to the usual closedness of the finite-dimensional replacement-channel set.  A physical seed at each access layer applies the boundary channel $\mathcal R_{I/d}$ once, or its complete-wire version after an enlarged encoding, and then uses idle system dynamics with an independent prepared environment.  Its finite-interface combs are compatible and its Born functional is continuous on the same instrument limits as idle propagation.  The erasure is a boundary CP map; a hardware model requiring its finite-time implementation includes that interval in the interface.  Thus $\Free_a$ is nonempty at every declared layer.  Arbitrary admissible prepared evolutions after entrance erasure are included among the seeds.  This free family is a specified transmission baseline, rather than the family of Markovian processes.

For a budget $B$, the class $\mathfrak C_{a\to\tau}(B)$ consists of admissible ideal-transparent physical control networks of the architecture below, with cost at most $B$ and output interface $a_\tau$.  Every such control obeys the full entrance-erasure cut of Definition~\ref{def:entrance-cut}, including its no-bypass clause: no input-dependent quantum system or classical record may cross the entrance on a side wire or be retained by the controller.  This condition is imposed on the control network itself, not inferred from preservation of the free set.  Initially independent ancillary supplies cannot carry the input around that cut.  The compatible references satisfy $P\Xi_a=\Xi_\tau$.  Each declared control subarchitecture is generated under identity, typed nesting, and independent classical randomization, with its hardware primitives and budget fixed before optimization.  A passive task restriction belongs to this class only when it preserves the prescribed ideal task.  Changing the permitted controls changes the optimization problem even when the task remains fixed.  Define the task-dependent feasible budget set by
\begin{equation}
 \mathcal B_{\tau,a}^{\mathrm{feas}}
 :=\{B:\mathfrak C_{a\to\tau}(B)\ne\varnothing\}.
 \label{eq:feasible-budget-domain}
\end{equation}
Task capability and its free baseline are defined only for $B\in\mathcal B_{\tau,a}^{\mathrm{feas}}$.  This set is upward closed under the budget order, but need not contain zero: a paid target gate can require positive control cost even on the ideal input process.  Identity conversions still have zero overhead; they are distinct from controls implementing a prescribed output task.  On the feasible budget set, the optimized free baseline is
\begin{equation}
 \eps_{\tau,a,B}^{\mathrm F}:=\inf_{\Lambda\in\Free_a}
 \inf_{P\in\mathfrak C_{a\to\tau}(B)}\eps_\tau(P\Lambda).
 \label{eq:free-baseline}
\end{equation}
Proposition~\ref{prop:replacement-baseline} computes this baseline exactly for unknown-state transmission.  For a different task, baseline positivity and compatibility under interface changes must be proved for that task.

\subsection{Physical free conversions and the entrance cut}

An implementation is a causal physical network with one open process slot, including the cPT's declared internal ports.  The controller can act on the main system wire, use initially independent ancillary states at accessible downstream ports, and use adapted downstream classical feedback when the subarchitecture permits it.  We impose the following condition on implementations, before passing to their operational equivalence classes.

\begin{definition}[Complete entrance-erasure cut]
\label{def:entrance-cut}
In each finite causal realization, unfold feedback into a directed acyclic network.  The protected entrance wire $M$ is a cut separating the external input $Q$ from every accessible output and surviving controller record.  All output wires of an upstream input encoding are routed into $M$ or discarded; no input-dependent quantum memory or classical record crosses the cut on a side wire.  Replacing the slot entrance by $\mathcal R^M_{\sigma_M}$ erases every component of $M$.  All other systems supplied across the cut, including the prepared environment, are initially independent of $QR$.  The tester-held reference $R$ is excluded from the controller network.  In addition, for every deterministic realization $\mathcal N$ and controller supply $\eta_C$ independent of $QR$, the replacement network must obey
\begin{equation}
 \mathcal N\!\left[\mathcal R^M_{\sigma_M}\right]
 (\rho_{QR}\otimes\eta_C)
 =\omega_{\mathcal N,\sigma_M,\eta_C}\otimes\rho_R,
 \qquad \rho_R=\Tr_Q\rho_{QR},
 \label{eq:cut-factorization}
\end{equation}
on all accessible output systems and complete controller records, before terminal measurement and without conditioning on a selected record.  A continuous realization is admitted through a declared instrument limit of networks satisfying these conditions, with Eq.~\eqref{eq:cut-factorization} retained on all accessible event laws.
\end{definition}

\begin{lemma}[No-signalling form of the complete entrance cut]
\label{lem:entrance-nosignalling}
Let an upstream deterministic encoding act on $Q$ and a controller supply $A$, producing the complete entrance wire $M$ and discarded output $D$.  If the supply is initially independent of $QR$, then, for every joint input state $\rho_{QR}$,
\begin{equation}
 \Tr_{MD}\!\left[(\mathcal E_{QA\to MD}\otimes\mathrm{id}_{RB})
 \Pi_{QR,AB}(\rho_{QR}\otimes\eta_{AB})\Pi_{QR,AB}^\dagger\right]
 =\rho_R\otimes\eta_B,
 \qquad \rho_R=\Tr_Q\rho_{QR}.
 \label{eq:entrance-nosignalling}
\end{equation}
After replacing the complete $M$ by $\mathcal R^M_{\sigma_M}$, the state across the cut is therefore $\sigma_M\otimes\rho_R\otimes\eta_B$.  The statement remains valid when $A$ and $B$ are initially correlated, provided the joint supply $AB$ is independent of $QR$ and all input-bearing encoder outputs are routed into $M$ or discarded.
\end{lemma}
\begin{proof}
Here $\Pi_{QR,AB}$ is the canonical unitary wire permutation from the order $QRAB$ to $QARB$, so $\mathcal E$ acts on $QA$ and the identity acts on $RB$.  Trace preservation of $\mathcal E$ gives the no-signalling identity after tracing its complete output, including the discarded register $D$.  Applying it yields Eq.~\eqref{eq:entrance-nosignalling}.  The complete replacement identity Eq.~\eqref{eq:complete-entrance-erasure} then replaces the full $M$ marginal while leaving $R$ and $B$ unchanged.  This establishes Eq.~\eqref{eq:cut-factorization} for the encoder followed by any deterministic downstream channel on $M$ and independent supplies.  No conditioning on an encoder outcome is used, so the conclusion also covers instruments after their complete record is included and summed.
\end{proof}

This definition excludes an encoding that swaps $Q$ into an ancillary wire subsequently spared by the replacement map.  A downstream ancilla may store a state after it has traversed the entrance; allowing such storage is a choice of subarchitecture, and does not restore the input after complete entrance erasure.

Composition means \emph{slot substitution}: insert the one-slot network for $S_1$ into the sole process slot of the network for $S_2$.  This nests two transformations around one original cPT, rather than serially running two copies of that cPT.  Only type-compatible substitutions preserving the entrance cut and admissible instrument limits are included in the generated architecture.

For a chosen subarchitecture, let $S:a\to a'$ be a deterministic causal conversion obtained from these primitives and their admissible continuous limits.  ``Deterministic'' means normalized, without postselection; a normalized random mixture is deterministic in this sense.  A conversion is free when:
\begin{enumerate}
\item Its physical realization is completely positive and causal under the declared ancillary extensions and satisfies the entrance-cut condition.
\item It preserves the compatible ideal task, $S\Xi_a=\Xi_{a'}$.  Randomized primitives used here satisfy this trajectory by trajectory.
\item Its instrument pullback $S^*$ preserves deterministic completion, outcome coarse-graining, causal prefixes, ancillary extension, and continuous instrument limits, and $S^*\Inst_{a'}\subseteq\Inst_a$.  This is the structural conversion condition; for a task control $P:a\to\tau$, the separate requirement is $P^*\Inst_\tau\subseteq\Inst_a$.
\end{enumerate}
Free-set preservation is a consequence of the architecture, proved below.  Here ``free'' denotes preservation of the specified free family and ideal reference; a free conversion can have a nonzero implementation cost.  That cost is debited from the budget attached to the object.  Seed preparation defines the baseline supply, while conversion acts on an already supplied process.  In the baseline below, the idle replacement seed is an available supply; its preparation is not charged as pulse action used to control that supplied process.

\subsection{Implementation costs and budget compatibility}

Costs are assigned to a chosen physical implementation $\kappa$ of a conversion, including its sampling and record hardware; we write $\Cost(\kappa)$, or $\Cost(S)$ when the implementation is fixed by context.  Different implementations of the same operational map may have different costs, and no minimum-cost implementation is assumed to exist.  Scalar budgets, or componentwise ordered vector budgets, obey
\begin{equation}
 \Cost(\id)=0,
 \qquad
 \Cost(S_2\circ S_1)\le \Cost(S_1)+\Cost(S_2).
 \label{eq:cost-axioms}
\end{equation}
For randomization we use a worst-branch budget, augmented by the cost of the sampler and any retained label.  An expected-cost convention is a separate architecture and is not substituted for a pathwise cap.  Section~\ref{sec:protocol-costs} gives concrete subadditive costs.  Identity has zero implementation overhead.  Composition is permitted only with matching interfaces, and capability is evaluated on the feasible budget set in Eq.~\eqref{eq:feasible-budget-domain}.  Thus ``free'' refers to structural admissibility and task preservation, while the selected physical implementation can carry a positive budget debit.

The budget domain contains zero and is closed under addition and admissible componentwise subtraction.  Control compatibility is imposed explicitly: for every admitted $S:a\to a'$ and $P\in\mathfrak C_{a'\to\tau}(B')$,
\begin{equation}
 P\circ S\in\mathfrak C_{a\to\tau}(B'+\Cost(S)).
 \label{eq:control-compatibility}
\end{equation}
We call the resulting structure a \emph{priced task resource theory}: its objects are triples $(\Lambda,a,B)$ with $B\in\mathcal B_{\tau,a}^{\mathrm{feas}}$, free conversions preserve the declared entrance cut and ideal task, and a conversion of implementation cost $c$ maps the budget to $B-c$ whenever $B-c\in\mathcal B_{\tau,a'}^{\mathrm{feas}}$.  Thus ``free'' refers to structural admissibility and target preservation; it does not mean that the physical implementation has zero hardware cost.
For generated physical networks this follows from slot substitution, ideal transparency, the entrance cut, and the implementation cost rule.  A selected parametrized protocol family is not automatically closed and is used as a source of feasible controls within that generated architecture.

\subsection{Contraction and closure}

\begin{lemma}[Physical pullback contraction]
\label{lem:physical-pullback}
For a physical conversion with the above pullback structure,
\begin{equation}
 d_{a'}(S\Lambda,S\Lambda')\le d_a(\Lambda,\Lambda').
 \label{eq:physical-contraction}
\end{equation}
It is therefore continuous in the operational topology.
\end{lemma}
\begin{proof}
For each output tester, the two output laws equal the laws of the same pulled-back input tester.  Taking the supremum over the pullback subset of $\Inst_a$ proves the inequality.  The same construction, at every ancillary extension and causal prefix, preserves positivity, normalization, and causality.  Compatibility with continuous limits makes it a conversion on the declared continuum space.
\end{proof}

\begin{proposition}[Identity and composition closure]
\label{prop:closure}
The identity conversion is free.  If $S_1:a_0\to a_1$ and $S_2:a_1\to a_2$ are free conversions, then $S_2\circ S_1:a_0\to a_2$ is free and obeys Eq.~\eqref{eq:cost-axioms}.
\end{proposition}
\begin{proof}
Identity preserves the entrance wire, reference, and ideal task and has the identity pullback.  Slot substitution leaves exactly one original process slot.  In the unfolded composite directed acyclic network, the routing clauses in the constituent cuts imply that every path carrying input dependence to an accessible output or surviving record crosses the single substituted complete entrance; a side route would contradict the cut condition in the constituent containing that route.  Replacing that entrance erases all its systems.  The no-signalling lemma handles an upstream deterministic encoder, and subsequent deterministic channels preserve the factorization in Eq.~\eqref{eq:cut-factorization}.  Thus the lemma supports the state factorization needed for closure, while the directed-graph routing argument establishes that the nested network still has a complete cut; neither step alone supplies both claims.  The entrance-erasure factorization also holds for correlated controller supplies whose joint state is initially independent of $QR$.  Such a correlated implementation is specified as one joint network; it need not equal the composition of the separately averaged maps.  Target preservation composes as $S_2S_1\Xi_{a_0}=\Xi_{a_2}$.  Pullbacks compose in reverse order, preserving ancillary extensions, outcomes, prefixes, and instrument limits.  Independent randomization uses the product law, pushed forward to the nested network.  Hardware costs obey Eq.~\eqref{eq:cost-axioms} for the declared physical composition.  Admissible operational limits inherit these event-law properties; no claim about an arbitrary limit of microscopic network graphs is required.
\end{proof}

\begin{proposition}[Free objects remain free]
\label{prop:free-preserve}
Every free conversion maps the closure in Eq.~\eqref{eq:free-set} into the corresponding output closure.
\end{proposition}
\begin{proof}
Insert an entrance-erasing seed into the protected slot of a finite conversion network.  The upstream encoding produces some state $\chi_{MR}$, including every input-bearing encoded component in $M$.  Equation~\eqref{eq:complete-entrance-erasure}, or equivalently Lemma~\ref{lem:entrance-nosignalling} followed by the replacement, sends it to $\sigma_M\otimes\rho_R$.  All remaining downstream supplies are independent of $QR$ as a joint supply, even if they are internally correlated.  A complete downstream channel acting on those supplies and $M$, while leaving $R$ untouched, therefore produces $\omega_{\mathcal L}\otimes\rho_R$.  Including all classical records in this channel proves the statement for adapted feedback without conditioning on a selected branch.  The resulting prepared dynamics belong to the generated free family at $a'$.  Finite free preparations and their mixtures obey the same argument.
For an admissible continuous conversion acting on a seed, the declared
instrument-limit compatibility retains this erasure factorization.
The output is an admissible prepared evolution after complete erasure,
which belongs to the seed family by its definition.
This step uses admissibility of the continuous realization.
For that fixed conversion, Lemma~\ref{lem:physical-pullback} then extends
the inclusion from generated free inputs to their $d_a$ closure.
At restricted interfaces the factorization means its accessible event-law equivalence class; at state-tomographic interfaces it is also an equality of accessible output states.
\end{proof}

\subsection{The task-capability monotone}

Fix the generated control architecture $\mathfrak C$ before defining capability; it is background data in every $E$ and $R$ below.  For a process $\Lambda\in\CPT_a$ and $B\in\mathcal B_{\tau,a}^{\mathrm{feas}}$, define the optimal task error, as an infimum at budget $B$, by
\begin{equation}
 E_{\tau,a,B}(\Lambda):=
 \inf_{\substack{P\in\mathfrak C_{a\to\tau}(B)}}
                    \eps_\tau(P\Lambda).
 \label{eq:budgeted-error}
\end{equation}
Smaller $E$ means greater task capability.  A resource-monotone convention with larger values corresponding to more task resource is
\begin{equation}
 R_{\tau,a,B}(\Lambda):=
 \left[\eps_{\tau,a,B}^{\mathrm F}-E_{\tau,a,B}(\Lambda)\right]_+.
 \label{eq:resource-monotone}
\end{equation}
For the same interface and budget, $R=0$ on every free object by Eq.~\eqref{eq:free-baseline}.  The task error $E$ remains the sole primary performance quantity.  The calibrated $R$ measures advantage over the best free performance; it need not be faithful, since a nonfree process can have no advantage for the chosen task.  We use infima throughout and assert attainment only when an explicit optimizer is supplied.  Convexity of the unoptimized error $\eps_\tau$ does not imply convexity of $E$, which also contains an infimum over controls; no such convexity is used here.

The budgeted objects are triples $(\Lambda,a,B)$, with free triples $\Lambda\in\Free_a$.  A conversion of cost $c$ maps such an object to $(S\Lambda,a',B-c)$; discarding unused budget is also permitted.  Identity acts at every $B$.  Two successive conversions costing $c_1,c_2$ leave $B-c_1-c_2$; their composite costs at most $c_1+c_2$, so any surplus is discarded to give the same output triple.  Thus the induced reachability relation is reflexive and transitive.  Equation~\eqref{eq:control-compatibility} supplies the associated control inclusion.  At fixed $B$, a priced conversion need not satisfy the same-budget inequality $R(S\Lambda)\le R(\Lambda)$, since that comparison supplies a new budget after paying for $S$.  For transmission the baseline is constant; nested control sets also give $E_{\tau,a,B_1}\ge E_{\tau,a,B_2}$ for $B_1\le B_2$, so budget discarding cannot increase $R$.  The triple order therefore has a resource monotone in the usual nonincreasing convention \cite{chitambar2019,berk2021}.

\begin{lemma}[Operational stability of optimized capability]
\label{lem:capability-stability}
For one architecture, task, and budget, if every admitted control $P\in\mathfrak C_{a\to\tau}(B)$ satisfies
\begin{equation*}
 P^*\Inst_\tau\subseteq\Inst_a,
\end{equation*}
then
\begin{equation}
 |E_{\tau,a,B}(\Lambda)-E_{\tau,a,B}(\Lambda')|
 \le d_a(\Lambda,\Lambda').
 \label{eq:capability-lipschitz}
\end{equation}
The same bound holds for $R$ with a common baseline.
\end{lemma}
\begin{proof}
The TV triangle inequality and pullback contraction bound the difference of the two task errors for every fixed control by $d_a(\Lambda,\Lambda')$.  Infimizing the two uniform inequalities gives Eq.~\eqref{eq:capability-lipschitz} without assuming an optimizer.  The map $x\mapsto[b-x]_+$ is 1-Lipschitz for fixed $b$.
\end{proof}

\begin{theorem}[Budgeted monotonicity]
\label{thm:monotonicity}
Let $S:a\to a'$ be a free conversion with cost $c=\Cost(S)$ and suppose the generated control architecture satisfies Eq.~\eqref{eq:control-compatibility}.  For $B\ge c$ with $B-c\in\mathcal B_{\tau,a'}^{\mathrm{feas}}$ (which implies $B\in\mathcal B_{\tau,a}^{\mathrm{feas}}$ by control compatibility),
\begin{equation}
 E_{\tau,a',B-c}(S\Lambda)\ge E_{\tau,a,B}(\Lambda).
 \label{eq:error-monotonicity}
\end{equation}
If the free baseline is preserved by the interface change, $\eps_{\tau,a',B-c}^{\mathrm F}=\eps_{\tau,a,B}^{\mathrm F}$, then
\begin{equation}
 R_{\tau,a',B-c}(S\Lambda)\le R_{\tau,a,B}(\Lambda).
 \label{eq:R-monotonicity}
\end{equation}
\end{theorem}
\begin{proof}
Take any admissible output control $P:a'\to a_\tau$ with $\Cost(P)\le B-c$.  Network closure, target preservation, and cost subadditivity make $P\circ S$ admissible at the input interface, with the same fixed output tests and cost at most $B$.  Therefore
\begin{equation}
 E_{\tau,a,B}(\Lambda)
 \le \eps_\tau(P S\Lambda).
 \end{equation}
Taking the infimum over $P$ gives Eq.~\eqref{eq:error-monotonicity}.  Substituting into Eq.~\eqref{eq:resource-monotone} yields Eq.~\eqref{eq:R-monotonicity} when the baseline is unchanged.
\end{proof}

Baseline equality in this theorem is substantive for a general task.  If it is not available, Eq.~\eqref{eq:error-monotonicity} remains the justified performance ordering; subtracting two moving baselines does not establish a resource monotone.  The following result supplies a common baseline for a concrete closed subtheory.

\begin{proposition}[Exact free transmission baseline]
\label{prop:replacement-baseline}
Let $\tau$ be transmission of an arbitrary $d$-dimensional input, with ideal endpoint channel $\mathrm{id}_d$, arbitrary reference-assisted input states, and arbitrary terminal POVMs.  Suppose each access layer admits the idle maximally mixed replacement seed and a zero-cost passive restriction to this task.  Under the entrance-cut architecture,
\begin{equation}
 \eps_{\tau,a,B}^{\mathrm F}=1-\frac1{d^2}
 \label{eq:exact-free-baseline}
\end{equation}
for every such $a$ and admissible $B$.  Thus $R_{\tau,a,B}=[1-d^{-2}-E_{\tau,a,B}]_+$ is monotone along all budgeted free conversions in this transmission subtheory.
\end{proposition}
\begin{proof}
Because every control in $\mathfrak C_{a\to\tau}(B)$ satisfies the entrance-cut/no-bypass condition, applying it to a free seed leaves the protected input inaccessible and produces an input-independent output state (with any downstream memory and complete record included in $\omega_{\mathcal L}$).  In the endpoint transmission task this is a replacement channel $\mathcal R_\sigma$ after the record is discarded.  Use the input
\begin{equation}
 |\Omega_d\rangle=\frac1{\sqrt d}\sum_{j=1}^d|j j\rangle.
\end{equation}
The ideal output is $|\Omega_d\rangle\langle\Omega_d|$, whereas replacement gives $\sigma\otimes I_R/d$.  The Bell event has probabilities $1$ and
\begin{equation}
 \langle\Omega_d|\sigma\otimes I_R/d|\Omega_d\rangle=d^{-2}.
\end{equation}
It therefore witnesses error at least $1-d^{-2}$ for every $\sigma$.  Retaining an input-independent classical controller label does not weaken this lower bound.

For the upper bound, the completely depolarizing channel $\mathcal D_d=\mathcal R_{I/d}$ is the uniform mixture of the $d^2$ Weyl unitary channels $\mathcal W_{pq}$, with $\mathcal W_{00}=\mathrm{id}_d$.  Weyl averaging and the triangle inequality give
\begin{equation}
 \frac12\|\mathcal D_d-\mathrm{id}_d\|_\diamond
 \le\frac1{d^2}\sum_{(p,q)\ne(0,0)}
 \frac12\|\operatorname{Ad}_{W_{pq}}-\mathrm{id}_d\|_\diamond
 \le1-d^{-2}.
\end{equation}
The second inequality uses the universal diamond bound for each term; for the
nonidentity Weyl operators the bound is attained because they are traceless
unitaries.  The Bell lower bound above applies to $\mathcal D_d$ itself, so
the displayed upper bound is tight:
\begin{equation}
 \frac12\|\mathcal D_d-\mathrm{id}_d\|_\diamond=1-d^{-2}.
 \label{eq:depolarizing-baseline-tight}
\end{equation}
The declared idle seed and passive control attain this value without control overhead.  Here ``zero overhead'' refers to the control budget: a hardware model may assign a separate preparation cost to the baseline seed.  Optimization over all free objects and admissible controls therefore gives Eq.~\eqref{eq:exact-free-baseline}, independently of interface and remaining budget.  The resource inequality follows from Theorem~\ref{thm:monotonicity}.
\end{proof}

The identity--depolarizing distance in Eq.~\eqref{eq:depolarizing-baseline-tight}
is a standard channel identity.  Its role here is to calibrate one free
baseline across the declared temporal interfaces and residual budgets.

\section{Exact temporal-access separation}
\label{sec:coarse}

\subsection{Coarse-graining as a free restriction}

Let $a_f\succeq a_c$ have the physical embedding $\iota:\Inst_{a_c}\to\Inst_{a_f}$ from Sec.~\ref{sec:setting}.  In the fine implementation of a coarse experiment, hidden ports are filled with specified deterministic operations, usually passive propagation, and hidden record variables are marginalized.  The coarse-graining map $G_{f\to c}$ is defined by
\begin{equation}
 p_{G_{f\to c}\Lambda}^{\mathcal J_c}:=
 p_\Lambda^{\iota(\mathcal J_c)},
 \qquad \mathcal J_c\in\Inst_{a_c}.
 \label{eq:coarse-map}
\end{equation}
The embedding preserves deterministic completion, ancillary extension, causal prefixes, outcome measures, and continuous instrument limits.  Hence $G$ is a physical conversion by Lemma~\ref{lem:physical-pullback}.  Choose compatible references $\Xi_{a_c}=G\Xi_{a_f}$.  Passive interface restriction adds no control hardware and is assigned zero overhead.  It respects the entrance cut and preserves the free family.  Both access structures must still support the same output task $\tau$; a restriction deleting a required task port is not a conversion within this task theory.  Coarse-graining forgets ports and records without resetting the environment or replacing a finite-duration instrument by independent channels.

The operational distance itself satisfies
\begin{equation}
 d_{a_c}(G\Lambda,G\Lambda')\le d_{a_f}(\Lambda,\Lambda'),
 \label{eq:coarse-distance-contraction}
\end{equation}
because the supremum is taken over a subset of experiments.  This contraction is an information statement, not a claim that coarse-graining improves the physical task.  The worst-case error can appear smaller simply because fewer questions are allowed.

\begin{proposition}[Matched-budget access ordering]
\label{thm:fine-advantage}
Fix one output task $\tau=(\Xi_\tau,\Inst_\tau)$ and one budget $B$.  Suppose for every $P_c\in\mathfrak C_{a_c\to\tau}(B)$ there is a lift $\widetilde P_c\in\mathfrak C_{a_f\to\tau}(B)$ satisfying
\begin{equation}
 p_{\widetilde P_c\Lambda}^{\mathcal J}
 =p_{P_cG\Lambda}^{\mathcal J}
 \quad\text{for all }\mathcal J\in\Inst_\tau,
 \qquad \Cost(\widetilde P_c)\le\Cost(P_c).
 \label{eq:task-lift}
\end{equation}
Then
\begin{equation}
 E_{\tau,a_f,B}(\Lambda)
 \le E_{\tau,a_c,B}(G_{f\to c}\Lambda).
 \label{eq:fine-advantage}
\end{equation}
For the zero-overhead passive restriction, the lift is $\widetilde P_c=P_c\circ G$.
\end{proposition}
\begin{proof}
Equation~\eqref{eq:task-lift} uses exactly the same task tester on both outputs, so the two candidates have equal task error.  Taking the infimum over fine controls, then over coarse controls, proves Eq.~\eqref{eq:fine-advantage}.  Composition closure and the zero cost of $G$ supply the stated passive lift.
\end{proof}

This structural ordering does not itself establish strictness.  It isolates the comparison that must be proved in a physical model.  If a restriction costs $c_G>0$, the same argument compares a fine budget $B+c_G$ with a coarse budget $B$; equal-budget advantage then needs a separately verified lift.  Increasing the output tester family is a different operation and is not used in this proposition.

\subsection{Exactly solvable two-window processes}
\label{sec:strict-model}

We now give a physical control class with an exactly solvable optimum.  Let the system $S$ and environment $E$ be qubits, with $H_E=0$ and initial environment $I_E/2$, independent of an arbitrary reference-assisted system input.  The protected interval consists of a coupling segment of length $h$, a control window of length $w$, a second coupling segment of length $h$, and a compensation window of length $w$, so $T=2h+2w$.  In this subsection $T$ denotes this protected duration; the total physical experiment also includes any paid boundary supports before and after it.  All generator costs below include those supports.  During the coupling segments, define
\begin{center}
\begin{tabular}{c|cc}
 & first interval & second interval\\
\hline
$A$ & $gZ_S\otimes Z_E$ & $gZ_S\otimes Z_E$\\
$B$ & $2gZ_S\otimes Z_E$ & $0$
\end{tabular}
\end{center}
and set the coupling to zero in the ideal control windows.  Let $a_c^{\rm end}$ denote the coarse endpoint interface whose only process-facing port is the input--output system wire; the two internal windows are closed and no intermediate tester insertion is admitted.  Write $\mathfrak G_{\rm end}$ for restriction to this interface.  With this interface and no control applied, both processes have the same accumulated phase $\phi=2gh$ and hence the same induced endpoint channel
\begin{equation}
 \mathfrak G_{\rm end}\Upsilon_A=\mathfrak G_{\rm end}\Upsilon_B=\mathcal D_\phi,
 \qquad
 \mathcal D_\phi[\rho]=\cos^2\!\phi\,\rho+\sin^2\!\phi\,Z\rho Z.
\label{eq:strict-coarse-channel}
\end{equation}
The phase addition uses the same persistent environment qubit, $H_E=0$, and conservation of $Z_E$, so the joint propagator is
\[
 e^{-ighZ_S\otimes Z_E}e^{-ighZ_S\otimes Z_E}=e^{-i2ghZ_S\otimes Z_E}.
\]
It is therefore not the composition of two channels whose environment has been reset between the coupling segments.
Under the declared $H_E=0$ model, the two processes have the same uncontrolled joint $S\otimes E$ propagator,
not merely the same reduced endpoint channel.  Their distinction is entirely
the temporal placement of the coupling relative to the accessible window.
At the coarse endpoint interface the two processes induce the same channel $\mathcal D_\phi$.  Consequently, every experiment in the declared coarse endpoint tester family $\Inst_{\tau}^{\rm end}\subseteq\Inst_{a_c^{\rm end}}\cap\Inst_{a_\tau}$, including arbitrary reference-assisted inputs and terminal POVMs, has the same output law for $A$ and $B$.  This is an operational equality at the endpoint interface.  It does not assert equality of the underlying multi-time process tensors if an intermediate window port is reopened.
For the identity task and arbitrary reference-assisted final measurements,
\begin{equation}
 \eps(\mathcal D_\phi)
 =\frac12\|\mathcal D_\phi-\mathrm{id}\|_\diamond
 =\sin^2\!\phi.
 \label{eq:strict-uncontrolled-error}
\end{equation}
Here $p=\sin^2(2gh)$ is the value for the explicitly declared model
$H_E=0$ and $\rho_E=I_E/2$.  A nonmaximally mixed environment or a nonzero
$H_E$ changes the reduced channel and hence the numerical function multiplying
the same temporal-access separation; the equal uncontrolled joint propagator
and the distinction created by opening the internal window remain the relevant
structural features and must then be reevaluated with the corresponding
environment model.

The output task is always the endpoint identity channel on the unknown input, tested with arbitrary input states on $SR$ and terminal POVMs; $R$ need only have dimension two for the diamond distance.  The tester does not insert arbitrary intermediate logical operations.  Paired pulses are choices of the controller, not restrictions on these endpoint tests.

For this endpoint example, the coarse access structure $a_c$ and its restriction map $G$ are the endpoint interface $a_c^{\rm end}$ and $\mathfrak G_{\rm end}$ just defined.  Thus the quantities $E_c(GX)$ below use exactly the endpoint tester family, whereas reopening an intermediate window would define a different access structure and a different optimization problem.

The fine subarchitecture permits system control supported on the two windows and a boundary encoding $W$ with inverse boundary decoding $W^\dagger$.  It uses no feedback or quantum storage and requires ideal transparency on each trajectory.  If the two window endpoints are $V_1,V_2$, the uncoupled endpoint channel is $\operatorname{Ad}_{W^\dagger V_2V_1W}$.  Target preservation on the ideal input therefore imposes $\operatorname{Ad}_{V_2V_1}=\mathrm{id}$, equivalently $V_2V_1=e^{i\theta}I$.  Thus the paired endpoints $V,V^\dagger$ are a consequence of the declared target-preservation architecture, rather than an additional claim about arbitrary controls.  Physical scalar phases are retained when computing generator costs; pairing refers to their induced channels.  If target preservation is dropped, independent window endpoints form a different control theory: for example, $V_1=I,V_2=Z$ changes the $B$ error to $1-\sin^2(2gh)$, as discussed below.  Boundary channels have the same hardware and counted costs at both layers, in identical coupling-free boundary supports outside the protected interval.  Taking $W=I$ has zero overhead.  Inside the protected interval require
\begin{equation}
 \operatorname*{ess\,sup}_{t}\|H_{\mathrm{ctrl}}(t)\|_\infty
 \le A_{\max},\qquad
 \int_0^T\|H_{\mathrm{ctrl}}(t)\|_\infty dt\le L_{\max}.
 \label{eq:window-budget}
\end{equation}
The common total budget is $B=(A_\infty,A_1)=(\Omega/2,\pi)$, counting action in all paid supports, including boundary supports.  Use $A_{\max}=\Omega/2$, $L_{\max}=\pi$, and $w=\pi/\Omega$; the optimizer uses $W=I$, so its boundary action is zero.  Thus the displayed $\pi$ is the total action of the two interior window pulses; if a hardware model assigns a cost to boundary supports, the same choice remains feasible and incurs no additional action.  These are the amplitude and action components of Eq.~\eqref{eq:protocol-cost}; other hardware supplies are fixed at both layers.  The idle maximally mixed replacement seed and passive task restriction are available with zero conversion overhead, as in Proposition~\ref{prop:replacement-baseline}.  The coarse interface closes both internal windows and allows only the boundary pair within the same budget.  The comparison fixes the same budget caps; each optimizer may use less than the available budget.  Requiring equal expenditure instead would define a different feasible set.  Independent mixtures are permitted with worst-branch pulse costs and a separately declared sampler supply.  The optima are attained deterministically and remain unchanged when those mixtures are included.

\subsection{Closed window-control architecture}

The window architecture retains its internal control ports until the final restriction to the endpoint task.  A frame module acts on an experiment attached to those ports, including any inner module's control field.  Let $G_{\rm app}(t)$ be the controllable apparatus Hamiltonian on $S$ and its instrument memory $A$.  Write $\widehat F=F\otimes I_{EA}$ and $\widehat K_F=K_F\otimes I_{EA}$, with $K_F=i\dot FF^\dagger$, and lift $G_{\rm app}$ to the same joint space.  The uncontrollable background is
\begin{equation}
 H_{\rm bg}(t)=I_S\otimes H_E(t)\otimes I_A
              +H_{SE}(t)\otimes I_A.
 \label{eq:window-background}
\end{equation}
Its laboratory Hamiltonian after inserting the frame module is
\begin{equation}
 H_{\rm lab}^{F,G_{\rm app}}(t)=H_{\rm bg}(t)+\mathcal A_F(G_{\rm app})(t),
 \qquad \mathcal A_F(G_{\rm app}):=\widehat K_F+\widehat F G_{\rm app}\widehat F^\dagger.
 \label{eq:frame-insertion-action}
\end{equation}
Thus $\mathcal A_F$ acts on the apparatus, while $H_{\rm bg}$ remains the same physical background.  The declared hardware permits arbitrary bounded integrable system control fields on the window supports, so the conjugated inner fields remain available.  When logical instruments are present, their physical implementation must admit the corresponding Hamiltonian form of pointwise covariance in Eq.~\eqref{eq:frame-insertion-action}, on the full retained instrument memory, as in Ref.~\cite{wang2026cprc}.  Source and target interfaces both retain these open ports, with laboratory and logical instrument descriptions respectively.  Endpoint restriction is applied after the frame modules have been composed.

\begin{lemma}[Closure of covariant window instrument pullbacks]
\label{lem:window-closure}
On this open-port architecture, let $F(t)$ be an absolutely continuous system frame, initially $I$, changing only on the declared windows and paid boundary supports, and finally a scalar phase times $I$.  Its plateaus are $W$, $VW$, and $e^{i\theta}W$ before the final boundary decoding.  Let $\mathcal T_F^*$ be the physical instrument pullback defined by Eq.~\eqref{eq:frame-insertion-action}, with the corresponding covariant outcome maps.  Then
\begin{equation}
 \mathcal T_{F_1}^*\mathcal T_{F_2}^*
 =\mathcal T_{F_1F_2}^*
 \label{eq:window-pullback-composition}
\end{equation}
as open instrument networks for the same fixed background.  The product is again a compensated window frame; its peak and action costs are at most the sums of the constituent costs.  Each module preserves the uncoupled logical experiment trajectory by trajectory.  The result concerns the declared apparatus pullback and its hardware closure, rather than a conjugation of the entire system--environment Hamiltonian.
\end{lemma}
\begin{proof}
First consider the physical insertion.  For $F=F_1F_2$, direct differentiation gives
\begin{equation}
 K=i\dot FF^\dagger=K_1+F_1K_2F_1^\dagger.
 \label{eq:window-frame-generator}
\end{equation}
The outer module covariantly implements the inner module's apparatus field, including $K_2$.  Hence the nested laboratory network has generator
\begin{align}
 H_{\rm bg}+\mathcal A_{F_1}(\mathcal A_{F_2}(G_{\rm app}))
 &=H_{\rm bg}+\widehat K_1
   +\widehat F_1\widehat K_2\widehat F_1^\dagger
   +\widehat{F_1F_2}G_{\rm app}\widehat{F_1F_2}^{\dagger}\\
 &=H_{\rm bg}+\mathcal A_{F_1F_2}(G_{\rm app}).
 \label{eq:frame-composition-identity}
\end{align}
This equality holds almost everywhere on the full physical interval, including boundary supports.  Uniqueness of the finite-dimensional propagator makes it an equality of the resulting joint evolutions for every apparatus input.  Outcome effects on $A$ and their record labels are retained.  If a block has monitored jumps, its jump operators are conjugated by $F\otimes I_A$; two successive conjugations give the same product frame.  Finite-duration kernels use their respective entrance and exit frames as in Eq.~\eqref{eq:finite-event-kernel-covariance}; the endpoint factors also multiply to $F_1F_2$, and frames at linked ports cancel.  Therefore every event kernel, including reference-assisted kernels, obeys Eq.~\eqref{eq:window-pullback-composition}.  All frames act trivially on the inaccessible environment, tester reference, and record registers.  This is nested insertion at open ports, not serial composition of already reduced endpoint channels.

Second, transforming this physical network to logical coordinates gives
\begin{equation}
 \widetilde H^{F,G_{\rm app}}
 =\widehat F^\dagger H_{\rm lab}^{F,G_{\rm app}}\widehat F
  -i\widehat F^\dagger\dot{\widehat F}
 =I_S\otimes H_E\otimes I_A
  +\widehat F^\dagger(H_{SE}\otimes I_A)\widehat F+G_{\rm app}.
 \label{eq:window-logical-generator}
\end{equation}
For $H_{SE}=0$, the generator and every logical outcome map are the prescribed ones.  With the same encoded input and covariant readout, each complete record law equals its ideal law.  This proves trajectory-wise ideal transparency by the same pointwise-covariance cancellation as in Ref.~\cite{wang2026cprc}.  For nonzero $H_{SE}$, the rotated coupling remains in Eq.~\eqref{eq:window-logical-generator}; it is this coupling that the echo modulates.  In particular, applying $\mathcal A_F$ to $H_{\rm bg}+G_{\rm app}$ instead of to $G_{\rm app}$ would describe a coordinate transformation of the entire evolution and would not establish this physical control claim.

Finally, the product frame has the same allowed derivative supports and scalar final value.  Its first and last plateaus are $W=W_1W_2$ and $e^{i(\theta_1+\theta_2)}W$, and its middle plateau is $VW$ with
\begin{equation}
 V=V^{(1)}W_1V^{(2)}W_1^\dagger.
 \label{eq:window-endpoint-product}
\end{equation}
Thus the internal window endpoints remain paired as channels.  Equation~\eqref{eq:window-frame-generator}, unitary invariance, and the triangle inequality bound the peak and integrated generator norms by the sums.  Independent randomization uses the product law pushed forward to $F_1F_2$ and satisfies the same worst-branch cost bound, with sampler costs added separately.  The implementation acts only on the declared system and apparatus ports and retains no input-bearing side wire, so it also respects the entrance cut.
\end{proof}

For the endpoint transmission task, set $G_{\rm app}=0$.  Then the actual laboratory dynamics are $H_{\rm bg}+\widehat K_F$, with the background coupling unchanged.  When $W=I$, the first coupling segment has frame $I$ and the second has frame $V$.  Choosing $V=X$ changes $ZZ$ to $-ZZ$ in Eq.~\eqref{eq:window-logical-generator}, producing the echo below.  In a nested implementation, the inner control field is $F_1K_2F_1^\dagger$, in addition to $K_1$; simple addition of uncompiled pulse fields is a different implementation.  The budget-graded family of open frame modules is closed, and a fixed-budget slice permits compositions whose combined cost fits that slice.  Coupling-free windows make the reference propagator depend only on these frame endpoints.  The exact optimum below therefore covers the physically nested modules and all admitted window waveforms.  With nonzero window coupling, closure of the apparatus modules persists, while the full propagator depends on the waveform and is controlled by the perturbation estimate below.

\subsection{Exact fine and coarse optima}

For clarity, write $\mathfrak C_{\mathrm{win}}$ for the fine class with $W=I$, and abbreviate
\[
 E_f(X):=E_{\tau,a_f,B}(X),\qquad E_c(GX):=E_{\tau,a_c,B}(GX).
\]
For this window architecture, $E_X^*$ denotes the optimal task error over the
admitted fine controls at the stated budget; hence
\[
 E_A^*=E_f(A),\qquad E_B^*=E_f(B).
\]
Its joint propagators are
\begin{equation}
 U_A(V)=V^\dagger e^{-ighZ_SZ_E}V e^{-ighZ_SZ_E},
 \qquad U_B(V)=e^{-i2ghZ_SZ_E}.
 \label{eq:window-propagators}
\end{equation}
Boundary pairs conjugate the corresponding reduced channels.  For $A$, the choice $V=X$ gives
\begin{equation}
 U_A=Xe^{-ighZ_SZ_E}Xe^{-ighZ_SZ_E}=I_{SE},
 \qquad E_A^*=0.
 \label{eq:strict-A-optimum}
\end{equation}
For process $B$, before imposing the paired target-preservation condition, the window endpoints give
\begin{equation}
 U_B(V_1,V_2)=V_2V_1 e^{-i2ghZ_SZ_E}.
 \label{eq:strict-B-general-endpoints}
\end{equation}
Within $\mathfrak C_{\mathrm{win}}$, target preservation imposes $V_2V_1=e^{i\theta}I$, and hence
\begin{equation}
 U_B=e^{i\theta}e^{-i2ghZ_SZ_E}\sim e^{-i2ghZ_SZ_E},
 \qquad E_B^*=\sin^2(2gh).
 \label{eq:strict-B-optimum}
\end{equation}
Consequently, in the notation of the endpoint interface just defined,
\begin{equation}
 \boxed{\mathfrak G_{\rm end}\Upsilon_A=\mathfrak G_{\rm end}\Upsilon_B,\qquad
 E_A^*=0,\qquad E_B^*=\sin^2(2gh).}
 \label{eq:strict-optimality}
\end{equation}
These values are exact over $\mathfrak C_{\mathrm{win}}$: errors are nonnegative and $A$ attains zero, while $B$ is independent of every window control.  The following theorem includes boundary controls and their randomization and makes the coarse-versus-fine comparison explicit.

\begin{theorem}[Strict temporal-access advantage at finite budget]
\label{thm:strict-access}
For the fixed endpoint identity task, the two processes above, and the paired-unitary subarchitectures with Eq.~\eqref{eq:window-budget}, let $p=\sin^2(2gh)$.  Then
\begin{equation}
 E_f(A)=0,\quad E_f(B)=p,\qquad
 E_c(GA)=E_c(GB)=p.
 \label{eq:strict-fine-coarse}
\end{equation}
Hence $p>0$ gives a strict access advantage for $A$ at the same budget, although $\mathfrak G_{\rm end}\Upsilon_A=\mathfrak G_{\rm end}\Upsilon_B$.  The free baseline is $3/4$ at both layers, and
\begin{equation}
 R_f(A)=\frac34,\quad
 R_f(B)=R_c(GA)=R_c(GB)=\left[\frac34-p\right]_+.
 \label{eq:strict-resource-values}
\end{equation}
\end{theorem}
\begin{proof}
With $W=I$, the admissible echo attains zero for $A$.  At the coarse endpoint interface, target preservation likewise restricts a general boundary pair to $W'=e^{i\theta}W^\dagger$, so every admissible deterministic coarse channel, and every fine channel for $B$, has the form
\begin{equation}
 (1-p)\mathrm{id}+p\mathcal N_W,
 \qquad \mathcal N_W(\rho)=N_W\rho N_W^\dagger,
 \quad N_W=W^\dagger ZW.
\end{equation}
Since $N_W$ is unitary and traceless, its Choi vector $(N_W\otimes I)|\Omega_2\rangle$ is orthogonal to $|\Omega_2\rangle$.  The Bell event therefore witnesses error $p$, even for an arbitrary random mixture of boundary pairs.  Convexity of the diamond norm gives the matching upper bound $p$.  This proves both coarse optima and the fine optimum for $B$.  For retained random labels, use the same Bell event while ignoring the label for the lower bound, and sum the branch upper bounds for the upper bound.  The finite pulse below attains the fine optimum for $A$.  Proposition~\ref{prop:replacement-baseline} supplies the baseline and Eq.~\eqref{eq:strict-resource-values}.
\end{proof}

The equality for the coarse optimum is a statement about the declared
paired-unitary boundary architecture, not about arbitrary CPTP pre- and
post-processing.  The strict $A$--$B$ endpoint-identification statement also
uses ideal transparency, which prevents the controller from changing the
logical task.  If transparency is relaxed, Eq.~\eqref{eq:relaxed-boundary-optimum}
gives the corresponding enlarged-class value.  The fine-versus-coarse
advantage for a single process is therefore a physical comparison within the
declared architecture, while the $A$--$B$ indistinguishability result is a sharper
statement for the transparent control class.

\subsection{An endpoint-only prediction limit}

\begin{proposition}[Endpoint-identification lower bound]
\label{prop:endpoint-prediction}
Let $\widehat E$ be any estimator whose input is only the uncontrolled endpoint
channel $\mathfrak G_{\rm end}\Upsilon$.  For the two processes in
Theorem~\ref{thm:strict-access}, with $p=\sin^2(2gh)$, one has
\begin{equation}
 \max\left\{
 \left|\widehat E(\mathfrak G_{\rm end}\Upsilon_A)-E_f(A)\right|,
 \left|\widehat E(\mathfrak G_{\rm end}\Upsilon_B)-E_f(B)\right|
 \right\}\ge \frac{p}{2}.
 \label{eq:endpoint-prediction-bound}
\end{equation}
More generally, if two physical realizations have a common endpoint channel and
their fine-access optima differ by $\Delta$, every endpoint-only estimator has
worst-case absolute error at least $\Delta/2$.
\end{proposition}
\begin{proof}
The endpoint channels of $A$ and $B$ are equal, so the estimator returns one
number, say $e$, on both inputs.  Since $E_f(A)=0$ and $E_f(B)=p$,
the triangle inequality gives $p\le |e|+|e-p|$.  The general statement is
identical with $p$ replaced by $\Delta$.  If an estimator uses independent
randomness or finite endpoint data, its output has the same law on both
processes.  Averaging the same triangle inequality then gives the bound on
the larger expected absolute error.  Thus even exact endpoint tomography
does not remove this identification error.
\end{proof}

\subsection{Fixed interaction phase and temporal allocation}

\begin{proposition}[Fixed total phase, variable temporal capability]
\label{prop:phase-allocation}
Retain the environment, coupling-free windows, endpoint task, and budget of
Theorem~\ref{thm:strict-access}.  For $0<\Phi\le\pi/2$ and
$0\le\lambda\le1$, let $\Upsilon_\lambda$ have respective coupling
Hamiltonians $(\lambda\Phi/h)Z_SZ_E$ and
$((1-\lambda)\Phi/h)Z_SZ_E$ on the two segments of duration $h$.
Every member has uncontrolled endpoint channel $\mathcal D_\Phi$, while
\begin{equation}
 E_c(\mathfrak G_{\rm end}\Upsilon_\lambda)=\sin^2\Phi,
 \qquad
 E_f(\Upsilon_\lambda)=\sin^2\!\bigl((2\lambda-1)\Phi\bigr).
 \label{eq:phase-allocation-optimum}
\end{equation}
These optima include boundary pairs and their admitted random mixtures.
At $\Phi=2gh$, $\lambda=1/2$ and $\lambda=1$ recover $A$ and $B$,
respectively, within the stated phase range.
\end{proposition}
\begin{proof}
Write $\alpha=\lambda\Phi$ and $\beta=(1-\lambda)\Phi$ for the two phases.
Their uncontrolled sum is $\Phi$, which proves the common endpoint channel
and its coarse optimum by the Bell witness in
Theorem~\ref{thm:strict-access}.
For a deterministic paired window control, let
$c=\Tr(V^\dagger ZVZ)/2\in[-1,1]$.  Conditional on the conserved
environment eigenvalue $e=\pm1$, the system propagator before boundary
conjugation is
\[
 U_e=e^{-ie\beta V^\dagger ZV}e^{-ie\alpha Z},\qquad
 q:=\frac12\Tr U_e=\cos\alpha\cos\beta-c\sin\alpha\sin\beta.
\]
Both environment branches have this same real normalized trace.
Since $\alpha,\beta\ge0$ and $\alpha+\beta\le\pi/2$,
\[
 0\le\cos(\alpha+\beta)\le q\le\cos(\alpha-\beta).
\]
For the maximally entangled input, the ideal Bell event has actual
probability $q^2$, also after any boundary conjugation $W$.
Its probability deficit therefore bounds every branch error below by
$1-q^2\ge\sin^2(\alpha-\beta)$.  The same event gives the same lower bound
after random mixing, including when the task retains control labels.
The feasible choice $V=-iX$, with its inverse in the compensation window,
gives $U_e=e^{-ie(\alpha-\beta)Z}$.  Averaging the environment yields
$\mathcal D_{\alpha-\beta}$, whose task error is $\sin^2(\alpha-\beta)$.
This attains the lower bound at the common peak and action caps.
\end{proof}

The family separates interaction strength from temporal accessibility:
the total phase and the maximum segment coupling $\Phi/h$ are uniformly
bounded, yet the fine optimum ranges from zero at balanced allocation to
$\sin^2\Phi$ at one-sided allocation.  At $\Phi=\pi/3$, all coarse
capability scores vanish, whereas the balanced member has
$R_f=3/4$.  Thus a common endpoint channel can conceal the full calibrated
transmission advantage.  Proposition~\ref{prop:endpoint-prediction}
gives a corresponding unavoidable prediction error of at least $3/8$.

\subsection{The exact echo-action threshold}

\begin{proposition}[Action threshold for a zero-error echo]
\label{prop:echo-action-threshold}
Within the system-only unitary window architecture of
Theorem~\ref{thm:strict-access}, a zero-error implementation for process $A$
with nontrivial coupling phase $\sin(gh)\cos(gh)\ne0$ requires, in each
coupling-free window, a net unitary $V$ satisfying
\begin{equation}
 V^\dagger ZV=-Z.
 \label{eq:echo-axis-reversal}
\end{equation}
If $A_1=\int\|H_{\rm ctrl}(t)\|_\infty dt$ is the action of one window and
$A_\infty=\operatorname*{ess\,sup}_t\|H_{\rm ctrl}(t)\|_\infty$ its peak
amplitude, then
\begin{equation}
 A_1\ge \frac{\pi}{2},\qquad A_\infty w\ge \frac{\pi}{2}
 \label{eq:echo-action-threshold}
\end{equation}
for each window.  Consequently, the paired two-window echo obeys
$A_1^{\rm total}\ge\pi$.  The rectangular pulse used in
Eq.~\eqref{eq:strict-pulse-budget} saturates these bounds, so they are both
necessary and sufficient within this declared architecture.
\end{proposition}
\begin{proof}
Set $\chi=gh$ and $c=\Tr(V^\dagger ZVZ)/2$.
As in Proposition~\ref{prop:phase-allocation}, the two environment
branches have normalized trace $q=\cos^2\chi-c\sin^2\chi$.
A zero-error identity channel has unit probability for the ideal Bell
event, so $q^2=1$.  Here $q\le1$ and
$q\ge\cos(2\chi)>-1$, because $\sin\chi\cos\chi\ne0$.
Hence $q=1$, which, since $\sin\chi\ne0$, forces $c=-1$.
For qubit Pauli axes this is precisely
Eq.~\eqref{eq:echo-axis-reversal}.  Boundary conjugations preserve $q$;
a random mixture with zero error must have zero Bell deficit in almost
every branch, so it satisfies the same condition branch by branch.
Condition~\eqref{eq:echo-axis-reversal} maps the two eigenspaces of $Z$ to
opposite eigenspaces.  Starting from either eigenstate, the unitary path in
that window must therefore reach an orthogonal state.  The Fubini--Study
distance between orthogonal rays is $\pi/2$, while the instantaneous speed
is the energy uncertainty
$\sqrt{\langle H_{\rm ctrl}^2\rangle-\langle H_{\rm ctrl}\rangle^2}
\le\|H_{\rm ctrl}(t)\|_\infty$.  Integrating the speed gives
$A_1\ge\pi/2$, and the peak bound gives $A_1\le A_\infty w$.
The second window implements $V^\dagger$ and has the same lower bound.
The pulse $H_{\rm ctrl}=(\Omega/2)X$ has action $\pi/2$ per window and
peak $\Omega/2$ over $w=\pi/\Omega$, attaining both inequalities.
At the excluded recurrence phases $\sin(2gh)=0$, the uncontrolled channel
is already the identity, so no positive control-action threshold applies.
\end{proof}

\subsection{The role of ideal transparency}

At $g=0.5$ and $h=1$, the exact resource values are $R_f(A)=0.75$ and $R_f(B)=R_c(GA)\simeq0.04193$.  The coarse channel thus fails to determine the resource available through its hidden temporal interface.  Boundary conjugations alone cannot remove the error in this declared control class.  Its ideal-transparency condition is substantive: if independent window endpoints are instead allowed, $(I,Z)$ implements
\begin{equation}
 \operatorname{Ad}_Z\circ\mathcal D_\phi
 =p\,\mathrm{id}+(1-p)\operatorname{Ad}_Z,
 \qquad \eps=1-p.
 \label{eq:nontransparent-correction}
\end{equation}
One $Z$ pulse has action $\pi/2$ and fits the stated peak cap in a window of length $\pi/\Omega$.  Thus for $p>1/2$ this enlarged control class improves $B$ below $p$ while pursuing the same identity task.  It violates ideal transparency; Theorem~\ref{thm:strict-access} asserts optimality in the compensated class, rather than across these distinct architectures.

For completeness, if the paired-endpoint constraint is dropped while retaining system-only unitary boundaries, the exact process-$B$ optimum is
\begin{equation}
 \inf_U\frac12\left\|\mathcal D_p-\operatorname{Ad}_U\right\|_\diamond
 =\min\{p,1-p\},
 \qquad \mathcal D_p=(1-p)\operatorname{id}+p\operatorname{Ad}_Z.
 \label{eq:relaxed-boundary-optimum}
\end{equation}
Unitary invariance of the diamond norm gives
\[
 \tfrac12\|\operatorname{Ad}_{V_2}\circ\mathcal D_p\circ
 \operatorname{Ad}_{V_1}-\mathrm{id}\|_\diamond
 =\tfrac12\|\mathcal D_p-
 \operatorname{Ad}_{(V_1V_2)^\dagger}\|_\diamond.
\]
Thus the optimization over independent unitary endpoints reduces to the
displayed optimization over $U$; $U=I$ and $U=Z$ are feasible at the stated
caps and give the two upper bounds.  For the converse, feed the normalized
Choi state into the diamond distance.  The Choi state of $\mathcal D_p$ is
$(1-p)|\Phi_I\rangle\langle\Phi_I|+p|\Phi_Z\rangle\langle\Phi_Z|$, with orthogonal Bell vectors.  For every normalized unitary Choi vector $|\Phi_U\rangle$, its overlap with this state is at most $\max\{p,1-p\}$, because the two displayed Bell components have squared overlaps summing to at most one.  The event $|\Phi_U\rangle\langle\Phi_U|$ therefore gives a trace-distance lower bound $1-\langle\Phi_U|\rho_{\mathcal D_p}|\Phi_U\rangle\ge\min\{p,1-p\}$.  The echo for $A$ remains exact in this enlarged class, so the separation persists for every $0<p<1$, while the values in Theorem~\ref{thm:strict-access} remain the values for the declared transparent class.

\subsection{Finite pulses and comparison protocols}

The separation survives finite control amplitude and finite pulse duration.  If
\begin{equation}
 \|H_{\mathrm{ctrl}}(t)\|_\infty\le\frac{\Omega}{2},
 \qquad w=\frac{\pi}{\Omega},
 \label{eq:strict-pulse-budget}
\end{equation}
then $H_{\mathrm{ctrl}}=(\Omega/2)X$ implements $V=-iX$ in the first window, with the reverse pulse implementing $V^\dagger$.  The peak norm is $\Omega/2$ and the total action of the pair is $\pi$, exactly the common budget caps in Eq.~\eqref{eq:window-budget}.  More generally, a bounded envelope with $\int_0^w\Omega_x(t)dt=\pi$ and no sign changes has the same endpoints and action.  Smooth envelopes require a longer window if the same peak cap is imposed.  No hard spectral-bandwidth claim is needed for this finite-amplitude construction.

For comparison, uniformly sampling $V\in\{I,X,Y,Z\}$ and applying $V^\dagger$ in the final window gives
\begin{equation}
 \overline{\mathcal E}_A=\frac12\,\mathrm{id}+\frac12\,\mathcal D_\phi,
 \qquad
 \eps(\overline{\mathcal E}_A)=\frac12\sin^2\!\phi,
 \qquad
 \eps(\overline{\mathcal E}_B)=\sin^2\!\phi.
 \label{eq:strict-random-control}
\end{equation}
Thus a universal random control improves process $A$, while a control selected using its internal time structure removes the error entirely.  The example is a benchmark for the resource quantity, not a claim about arbitrary continuous controls.

If the Pauli draw is retained as a flag, its law is also included in the ideal output.  For a flagged difference $\sum_r p_r|r\rangle\langle r|\otimes\Delta_r$, the diamond norm is the supremum, over a common input state, of $\sum_r p_r\|(\Delta_r\otimes\mathrm{id})(\rho)\|_1$.  Here two branches have zero difference and two share $\mathcal D_\phi-\mathrm{id}$, so the same error $p/2$ holds with or without the flag.  This observation is specific to these branches, not a maximum-over-branches formula.

As a boundary case, isotropic Markov depolarization obeys
\begin{equation}
 \dot\rho=-i[H_{\mathrm{ctrl}}(t),\rho]
 +\gamma(t)\left(\frac I2\Tr\rho-\rho\right),
 \qquad
 \eps=\frac34\left(1-e^{-\int_0^T\gamma(t)dt}\right),
 \label{eq:isotropic-boundary}
\end{equation}
Here $\gamma(t)\ge0$ and the total control unitary has been compensated.  The dissipative generator commutes with every system rotation, so the resulting channel is $e^{-\int\gamma}\,\mathrm{id}+(1-e^{-\int\gamma})\mathcal R_{I/2}$ irrespective of the system-unitary waveform.  The displayed error therefore applies to compensated system-unitary control, without extending the claim to auxiliary-assisted error correction.  Internal access improves the task when the admitted controls can change error accumulation, as in Eqs.~\eqref{eq:strict-optimality}--\eqref{eq:strict-robust-gap}.

\section{Task Covariance and Extensions Beyond Identity Transmission}
\label{sec:task-covariance}

The temporal-access question applies to logical gates as well as transmission. The main result of this section is a direct gate implementation: the same two qubit processes retain the exact $0/p/p$ separation for every $U(2)$ target at a sufficient common budget. Output-task covariance then identifies the error in a different logical frame, while input--output transport specifies the additional control and cost correspondence needed for an abstract equality of optimized errors.

\subsection{Unitary tasks and fixed-process output covariance}

Let $S$ have dimension $d$ and let $U$ be a unitary on $S$.  Write
\begin{equation}
 \Xi_U:=\operatorname{Ad}_U,\qquad
 \operatorname{Ad}_U(\rho):=U\rho U^\dagger,
 \label{eq:unitary-task}
\end{equation}
and let $\tau_{\rm id}=(\mathrm{id},\Inst_{\rm id})$ denote the identity
transmission task with its fixed reference-assisted complete-record tests.
The unitary task $\tau_U=(\Xi_U,\Inst_U)$ is defined by transporting every
identity test through the same output frame.  More explicitly, there is a
bijective measurable correspondence $\mathcal J\leftrightarrow\mathcal J_U$
between the two declared test families such that, for every actual process
$\Lambda$,
\begin{equation}
 \TV\!\left(p_\Lambda^{\mathcal J_U},
 p_{\Xi_U}^{\mathcal J_U}\right)
 =
 \TV\!\left(p_{\Gamma_U(\Lambda)}^{\mathcal J},
 p_{\mathrm{id}}^{\mathcal J}\right),
 \qquad
 \Gamma_U(\Lambda):=\operatorname{Ad}_{U^\dagger}\circ\Lambda .
 \label{eq:unitary-test-covariance}
\end{equation}
For the endpoint transmission tests used below this correspondence is
explicit.  If $\mathcal J=(\rho_{QR},\{M_x\}_x)$, with the reference factor
of each effect left unchanged, then
\begin{equation}
 \mathcal J_U=(\rho_{QR},\{(U\otimes I_R)M_x(U^\dagger\otimes I_R)\}_x).
 \label{eq:unitary-test-map}
\end{equation}
Thus the input state and reference readout are fixed while only the terminal
output frame is transported.  For every outcome,
\(\Tr[(U M_xU^\dagger)\Lambda(\rho)]
 =\Tr[M_x\,\Gamma_U(\Lambda)(\rho)]\), and the ideal side gives
\(\Tr[(U M_xU^\dagger)\operatorname{Ad}_U(\rho)] =\Tr[M_x\rho]\).
The same output-frame map is applied to the terminal factor of a multi-time
tester, with all earlier instruments and record labels unchanged.
The correspondence includes the reference and terminal readout, and uses the
same complete record space.  Thus a frame is not silently added to the task
after optimization.  If the interface supplies $U$ as a logical output frame,
its conversion cost is $c_U=0$.  Otherwise $c_U$ is the declared cost of a
physical implementation of this frame, including any required boundary
supports and task-interface hardware.

\begin{proposition}[Fixed-process output-frame covariance]
\label{prop:fixed-output-covariance}
Assume the output-frame correspondence in Eq.~\eqref{eq:unitary-test-covariance}
preserves the entrance cut and maps the declared control classes as follows:
for every budget $B\in\mathcal B_{\tau_U,a}^{\mathrm{feas}}$ with $B-c_U\in\mathcal B_{\tau_{\rm id},a}^{\mathrm{feas}}$, composing a control for $\tau_U$ with the
inverse output frame gives an admissible control for $\tau_{\rm id}$ at budget
$B-c_U$, and conversely every identity-task control at budget $B-c_U$ lifts
with the same certified cost decomposition.  Then
\begin{equation}
 \boxed{
 \eps_{\tau_U}(\Lambda)
 =\eps_{\tau_{\rm id}}\!\left(\Gamma_U(\Lambda)\right),\qquad
 E_{\tau_U,a,B}(\Lambda)
 =E_{\tau_{\rm id},a,B-c_U}(\Lambda).}
 \label{eq:fixed-output-capability-covariance}
\end{equation}
If the free family is covariant under the same frame map, then
\begin{equation}
 \eps^{\mathrm F}_{\tau_U,a,B}
 =\eps^{\mathrm F}_{\tau_{\rm id},a,B-c_U}.
 \label{eq:unitary-baseline-covariance}
\end{equation}
In particular, when $c_U=0$ the calibrated advantages of the two tasks are
identical after the same output-frame relabeling.
\end{proposition}
\begin{proof}
For each paired test, Eq.~\eqref{eq:unitary-test-covariance} gives equality
of the two TV distances.  The correspondence is bijective, so taking the
supremum proves the first equality in Eq.~\eqref{eq:fixed-output-capability-covariance}.
The assumed control correspondence pairs every feasible $P$ on the left with
a feasible identity-task control on the right, with exactly the residual
budget $B-c_U$, and vice versa.  Taking the infimum over the paired controls
proves the second equality.  The process on the right remains $\Lambda$;
only the output task and the control class have changed; the $\Gamma_U$ in the
error identity is absorbed into this paired control correspondence.  Applying
the same argument to the free family
proves Eq.~\eqref{eq:unitary-baseline-covariance}.  The exact budget shift
requires this two-sided, control-independent cost decomposition.  When a
target gate and noise control share an implementation, their costs can merge;
then the equality is replaced by the directly certified feasible controls or
the applicable one-sided bound.
\end{proof}

\subsection{Input--output and control transport}

\begin{theorem}[Input--output transport covariance]
\label{thm:unitary-task-covariance}
Let $F_{\rm in}:\CPT_{a_0}\to\CPT_{a_U}$ transport an input process and its
compatible ideal reference, and let $F_{\rm out}$ transport the output task,
with $F_{\rm out}\Xi_{\rm id}=\Xi_U$.  Assume the task experiments are
bijective under the corresponding input and output frame maps.  Suppose the
control classes are related by the two-sided intertwining relation
\begin{equation}
 P_U\circ F_{\rm in}=F_{\rm out}\circ P_0,
 \qquad
 P_0\in\mathfrak C_{a_0\to\tau_{\rm id}}(B),\quad
 P_U\in\mathfrak C_{a_U\to\tau_U}(B+c_U),
 \label{eq:unitary-control-intertwining}
\end{equation}
where both budgets belong to their respective feasible sets, every admissible control has a counterpart, and the displayed cost shift
is an exact certified decomposition, meaning
\(\Cost(P_U)=\Cost(P_0)+c_U\) for the paired implementations.  Then
\begin{equation}
 \boxed{
 E_{\tau_U,a_U,B+c_U}\!\left(F_{\rm in}\Lambda\right)
 =E_{\tau_{\rm id},a_0,B}(\Lambda).}
 \label{eq:unitary-transport-capability}
\end{equation}
If every identity-task control lifts with cost at most $c_U$ but the reverse
map is not available, the conclusion is only
\[
 E_{\tau_U,a_U,B+c_U}(F_{\rm in}\Lambda)
 \le E_{\tau_{\rm id},a_0,B}(\Lambda);
\]
the reverse one-sided control map gives the reverse inequality.  Exact
equality requires both directions.
When $c_U=0$, the transported tasks have identical optimized errors and,
provided the free family is transported as well, identical calibrated
advantages.
\end{theorem}
\begin{proof}
For a paired control satisfying Eq.~\eqref{eq:unitary-control-intertwining},
the output experiment correspondence gives
\[
 \eps_{\tau_U}\!\left(P_UF_{\rm in}\Lambda\right)
 =\eps_{\tau_{\rm id}}(P_0\Lambda).
\]
The two-sided, exact-budget correspondence identifies the feasible control
sets after the shift $B\mapsto B+c_U$.  Taking infima proves
Eq.~\eqref{eq:unitary-transport-capability}; applying the same argument to
the transported free family proves the corresponding baseline statement.
\end{proof}

\subsection{Exact free baseline for unitary transmission}

The entrance-erasing architecture has a task-independent transmission
baseline along the entire unitary orbit.

\begin{proposition}[Unitary free baseline]
\label{prop:unitary-free-baseline}
For the unknown $d$-dimensional input task with complete reference-assisted
endpoint tests, suppose $B\in\mathcal B_{\tau_U,a}^{\mathrm{feas}}$ and a
free process together with a budget-$B$ target-transparent control realizes
the completely depolarizing endpoint channel.  Then every same-dimensional unitary target satisfies
\begin{equation}
 \boxed{\eps^{\mathrm F}_{\tau_U,a,B}=1-d^{-2}.}
 \label{eq:unitary-free-baseline}
\end{equation}
The stated achievability condition holds, for example, when the idle
maximally mixed replacement seed is supplied and a target-transparent
control fits the budget while leaving its endpoint output depolarizing.
\end{proposition}
\begin{proof}
Every complete entrance-erasing output is a replacement channel.  Let
$|\Omega_d\rangle=d^{-1/2}\sum_j|j\rangle|j\rangle$ and
$|\Omega_U\rangle=(U\otimes I)|\Omega_d\rangle$.  For any replacement state
$\sigma$, the Bell event associated with $|\Omega_U\rangle$ has probability
\begin{equation*}
 \langle\Omega_U|\left(\sigma\otimes I/d\right)|\Omega_U\rangle=d^{-2},
\end{equation*}
so the TV error is at least $1-d^{-2}$.  The completely depolarizing channel
$\mathcal D_d$ is a uniform mixture of $d^2$ Weyl conjugations and obeys
$\frac12\|\mathcal D_d-\mathrm{id}\|_\diamond=1-d^{-2}$.  By unitary invariance,
$\frac12\|\mathcal D_d-\operatorname{Ad}_U\|_\diamond$ has the same value.
 The assumed budget-feasible free process and target-transparent control
 attain this value.  In the idle-seed realization, the depolarizing output
 is invariant under the unitary output frame, and its diamond distance to
 $\operatorname{Ad}_U$ equals its distance to the identity.
 This proves the equality.
\end{proof}

\subsection{Direct physical separation for qubit unitary gates}

\begin{proposition}[Direct qubit realization for every unitary target]
\label{prop:direct-unitary-qubit}
For the two-qubit processes of Sec.~\ref{sec:strict-model}, let the target be
$\Xi_U=\operatorname{Ad}_U$ with $U\in U(2)$.  Let $W$ and $L$ be the
boundary encoding and decoding unitaries and let $V_1,V_2$ be the two window
endpoints.  Impose trajectory-wise ideal transparency in the form
\begin{equation}
 L V_2V_1W=e^{i\theta}U .
 \label{eq:unitary-qubit-transparency}
\end{equation}
Put $p=\sin^2(2gh)$, with the environment initially maximally mixed and
$H_E=0$.  For every allowed branch satisfying Eq.~\eqref{eq:unitary-qubit-transparency},
the process $B$ has controlled joint propagator
\begin{equation}
 U_{B,\mathrm{ctrl}}^{(U)}
 =L V_2V_1 e^{-i2gh Z_S\otimes Z_E}W
 =e^{i\theta}U W^\dagger e^{-i2gh Z_S\otimes Z_E}W,
 \label{eq:unitary-qubit-B-propagator}
\end{equation}
and hence reduced channel
\begin{equation}
 \mathcal E_{B}^{(U)}
 =\operatorname{Ad}_U\circ\left[(1-p)\,\mathrm{id}
   +p\,\operatorname{Ad}_{W^\dagger ZW}\right].
 \label{eq:unitary-qubit-B-channel}
\end{equation}
Its error relative to $\operatorname{Ad}_U$ is exactly $p$.  The same Bell
event gives the lower bound $p$ for every admissible random mixture of such
branches, so the value remains exact after randomization.

For process $A$, the choice
\begin{equation}
 W=I,\qquad V_1=-iX,\qquad V_2=iX,\qquad L=U
 \label{eq:unitary-qubit-A-control}
\end{equation}
gives
\begin{equation}
 U_{A,\mathrm{ctrl}}^{(U)}
 =U(iX)e^{-igh Z_S\otimes Z_E}(-iX)e^{-igh Z_S\otimes Z_E}
 =U.
 \label{eq:unitary-qubit-A-propagator}
\end{equation}
Assume that the fine and coarse target-$U$ architectures include the same
boundary pair and terminal gate, forbid input-bearing side wires, and admit
the displayed pulses whenever their declared costs fit a common budget $B$.
If the two finite pulses and the boundary gate $L=U$ fit that budget, then,
at the same budget,
\begin{equation}
 \boxed{
 E_{\tau_U,f,B}(\Upsilon_A)=0,\qquad
 E_{\tau_U,f,B}(\Upsilon_B)=p,\qquad
 E_{\tau_U,c,B}(\mathfrak G_{\rm end}\Upsilon_A)=p .}
 \label{eq:unitary-qubit-direct-separation}
\end{equation}
In the peak--action accounting of Eq.~\eqref{eq:window-budget}, the two
interior pulses contribute peak $\Omega/2$ and total action $\pi$; the
implementation cost of the boundary gate $U$ is added according to the
declared hardware model.  Equivalently, the displayed result applies whenever
the certified total cost of $L=U$ together with these two pulses is at most
$B$.  This is a sufficient budget certificate; joint gate-and-echo synthesis
can cost less, as shown below.
Here the first equality is attained by Eq.~\eqref{eq:unitary-qubit-A-control},
whereas the latter two follow from Eq.~\eqref{eq:unitary-qubit-B-channel}
and the same rotated Bell witness.  This is a direct physical statement for
the original processes $A$ and $B$ in the two-dimensional model; it does not
assume that all unitary-task control classes admit a universal exact budget
translation.  The general $d$-dimensional free baseline in
Proposition~\ref{prop:unitary-free-baseline} remains valid, while this
explicit echo construction establishes the target orbit $U(2)$.
\end{proposition}
\begin{proof}
The process-$B$ expression follows because both coupling segments precede the
window controls, so their joint propagator is $e^{-i2gh Z_S\otimes Z_E}$.
Solving Eq.~\eqref{eq:unitary-qubit-transparency} for $LV_2V_1$ gives the
second equality in Eq.~\eqref{eq:unitary-qubit-B-propagator}.  Tracing the
maximally mixed environment then yields Eq.~\eqref{eq:unitary-qubit-B-channel}.
Since $W^\dagger ZW$ is a traceless qubit unitary, its Choi vector is
orthogonal to the identity Bell vector.  The Bell event associated with the
target $U$ therefore has probability $1-p$ for every branch, giving the
lower bound $p$; the diamond-norm triangle bound gives the matching upper
bound.  Linearity preserves this witness for a random mixture.

For process $A$, the chronological propagator is
$L V_2e^{-igh Z_S\otimes Z_E}V_1e^{-igh Z_S\otimes Z_E}W$.
Substitution of Eq.~\eqref{eq:unitary-qubit-A-control} and
$Xe^{-ighZZ}Xe^{-ighZZ}=I$ gives Eq.~\eqref{eq:unitary-qubit-A-propagator}.
The ideal-transparency condition is also satisfied because
$LV_2V_1W=U$.  Thus the displayed controls are feasible whenever their
declared pulse, boundary, and hardware costs fit $B$, proving the three
equalities in Eq.~\eqref{eq:unitary-qubit-direct-separation}.
\end{proof}

The target $U=-iX$ gives a sharper implementation within the same system-only
window architecture.  Set $W=L=V_2=I$ and apply only $V_1=-iX$ in the first
coupling-free window.  Ideal transparency now reads $LV_2V_1W=-iX$, and
\begin{equation}
 U_{A,\mathrm{ctrl}}^{(-iX)}
 =e^{-igh Z_S\otimes Z_E}(-iX)e^{-igh Z_S\otimes Z_E}
 =-iX .
 \label{eq:merged-gate-echo}
\end{equation}
The pulse takes duration $w=\pi/\Omega$, peak amplitude $\Omega/2$, and total
action $\pi/2$.  By Proposition~\ref{prop:echo-action-threshold}, perfect
identity transmission for $A$ in the paired-window architecture requires
total action $\pi$; separately implementing $-iX$ with the same pulse
hardware costs another $\pi/2$.  The single pulse in
Eq.~\eqref{eq:merged-gate-echo} performs both the gate and the echo with
action $\pi/2$.  Process $B$ retains error $p$ by
Eq.~\eqref{eq:unitary-qubit-B-channel}.  This exact implementation shows why
an additive, control-independent gate charge $c_U$ is an extra hypothesis in
the covariance statements above; physical gate and protection controls can
share the same pulse.

\subsection{Transporting the temporal-access separation}

The exact finite-window result is therefore a statement about a whole orbit of
logical tasks.  For the processes of Sec.~\ref{sec:strict-model}, define
\begin{equation}
 \Upsilon_A^{(U)}:=\operatorname{Ad}_U\circ\Upsilon_A,\qquad
 \Upsilon_B^{(U)}:=\operatorname{Ad}_U\circ\Upsilon_B,
 \label{eq:unitary-lifted-processes}
\end{equation}
with the same output frame applied to the ideal reference.  Then
$\Gamma_U(\Upsilon_X^{(U)})=\Upsilon_X$.  When the input and output frames
obey the two-sided intertwining condition of Theorem~\ref{thm:unitary-task-covariance},
the transported result is, for a source budget $B$,
\begin{equation}
 E_{f,B+c_U}^{(U)}(\Upsilon_A^{(U)})=0,\qquad
 E_{f,B+c_U}^{(U)}(\Upsilon_B^{(U)})=p,\qquad
 E_{c,B+c_U}^{(U)}(\mathfrak G_{\rm end}\Upsilon_A^{(U)})=p,
 \quad p=\sin^2(2gh).
 \label{eq:unitary-lifted-separation}
\end{equation}
Here $E_{f,B}^{(U)}:=E_{\tau_U,a_f,B}$ and
$E_{c,B}^{(U)}:=E_{\tau_U,a_c,B}$ denote the fine and coarse optima at the
indicated budget.  If only a zero-cost output frame is supplied while the
input process is held fixed, the fixed-process proposition applies instead,
and no $\Gamma_U$ is applied to the input of the optimized error.  For a
nonzero frame cost, the displayed equality requires the certified exact cost
decomposition; general subadditivity gives only the corresponding feasible
control bound.
The endpoint equality and the prediction lower bound are unchanged because
both processes are transformed by the same frame.  Thus the temporal
separation is a property of the internal access structure, rather than a
special feature of the identity label for the logical task.

\subsection{Scope for nonunitary tasks}

For a general CPTP gate, channel-simulation task, measurement, or multi-time
logical process, the definitions of $\eps_\tau$, $E_{\tau,a,B}$, and the
budgeted conversion order apply on the task's feasible budget set.  A task-specific output-frame map
can be used whenever it gives a bijection of complete experiment families and
preserves the entrance cut.  The free baseline must then be recomputed for the
new task and checked for compatibility across access layers.  If the baseline
is not preserved, Theorem~\ref{thm:monotonicity} still gives the ordering of
the optimized error $E$, but it does not by itself give monotonicity of the
calibrated advantage $R$.  CPRC performance certificates transfer to such a
task only when its uncoupled reference, instrument covariance, initial
independence, autonomous control law, and complete-record convention satisfy
the hypotheses stated in Sec.~\ref{sec:cprc}.  Each such task specifies its own free benchmark and control architecture. The unitary transmission family admits the common exact benchmark above and the stated transport correspondences.

\section{Task-independent implementation stability}
\label{sec:implementation-stability}

The complete-record metric also compares a nominal conversion with its physical implementation. This comparison is meaningful for any fixed logical task and includes controller, pulse, instrument, memory, and record errors. Environmental dynamics already assigned to the cPT remain inside the process object. Uniform implementation estimates transfer directly to optimized task error, including the unitary targets of Sec.~\ref{sec:task-covariance} when their physical costs fit the declared budget.

\subsection{Complete-record implementation discrepancies}

\begin{proposition}[Implementation discrepancy]
\label{prop:implementation-stability}
Let $P$ be a certified ideal implementation and $\widetilde P$ a physical implementation with the same logical output record space.  Define
\begin{equation}
 \delta_\tau(\widetilde P,P;\Lambda)
 :=\sup_{\mathcal J\in\Inst_\tau}
 \TV\!\left(p_{\widetilde P\Lambda}^{\mathcal J},
 p_{P\Lambda}^{\mathcal J}\right).
\end{equation}
If $\eps_\tau(P\Lambda)\le u_P$, then
\begin{equation}
 \eps_\tau(\widetilde P\Lambda)
 \le\min\{1,u_P+\delta_\tau(\widetilde P,P;\Lambda)\}.
 \label{eq:implementation-certificate}
\end{equation}
For a finite network of full record-preserving CPTP blocks, half-diamond discrepancies $\delta_j$ for corresponding blocks give $\delta_\tau\le\sum_j\delta_j$, uniformly over the declared task testers, provided all blocks retain every system carrying correlations across a block boundary, including environment and instrument memories and records.  A local control block may treat untouched systems as diamond-norm references.  For a block containing system--environment evolution, the environment output is retained for the next block; its reduced system endpoint channel is not substituted for that full block.  Internal insertions use the open-strategy or retained-generator topology of Proposition~\ref{prop:regular-pullback}.
\end{proposition}
\begin{proof}
The triangle inequality for each common record law, followed by the same supremum over $\Inst_\tau$, proves Eq.~\eqref{eq:implementation-certificate}.  Replace the blocks one at a time.  Trace-distance contraction by each preceding and subsequent complete channel and the reference-assisted diamond bound give $\delta_j$ per replacement, and telescoping gives the sum.  For adaptive blocks, the bound is required uniformly in the past record; they act as complete CPTP maps on the joint record register.  This argument never normalizes a selected measurement branch.
\end{proof}

A uniform Hamiltonian implementation perturbation on a bounded joint dilation can instead be estimated by the Duhamel bound proved below.  A realized control belongs to the exact resource conversion class only if it still satisfies the declared target, cut, and budget requirements.  Otherwise Eq.~\eqref{eq:implementation-certificate} is a performance certificate for that realization, rather than a claim of free conversion.  It separates control imperfections from the cited CPRC dynamical error without extending that theorem's autonomy or covariance assumptions.

\begin{lemma}[Stability of budgeted task optima under paired implementations]
\label{lem:paired-implementation-stability}
Fix a task $\tau$ and a budgeted control class $\mathfrak C$. Suppose every nominal implementation $P\in\mathfrak C$ has a corresponding physical implementation $\widetilde P$ with the same logical record interface and a uniform complete-record discrepancy at most $r$. Define $E^{\rm nom}:=\inf_{P\in\mathfrak C}\eps_\tau(P\Lambda)$ and $E^{\rm impl}:=\inf_{P\in\mathfrak C}\eps_\tau(\widetilde P\Lambda)$. Then
\begin{equation}
 |E^{\rm impl}-E^{\rm nom}|\le r.
 \label{eq:paired-implementation-stability}
\end{equation}
The statement compares the paired implemented controls. It does not assert that an imperfect implementation belongs to the exact free-conversion class.
\end{lemma}
\begin{proof}
The TV triangle inequality gives $|\eps_\tau(\widetilde P\Lambda)-\eps_\tau(P\Lambda)|\le r$ for every pair. Taking an infimum on both sides in each direction proves the claim.
\end{proof}

\subsection{Coherent implementation bounds and process-model stability}

\begin{lemma}[Unitary-channel perturbation bound]
\label{lem:unitary-channel-perturbation}
For finite-dimensional unitaries $U$ and $V$,
\begin{equation}
 \frac12\left\|\operatorname{Ad}_U-\operatorname{Ad}_V\right\|_\diamond
 \le \|U-V\|_\infty .
 \label{eq:unitary-channel-perturbation}
\end{equation}
If $U(t)$ and $V(t)$ are generated by bounded Hamiltonians $H(t)$ and $\widetilde H(t)$, then
\begin{equation}
 \frac12\left\|\operatorname{Ad}_{U(T)}-\operatorname{Ad}_{V(T)}\right\|_\diamond
 \le \int_0^T\|H(t)-\widetilde H(t)\|_\infty dt .
 \label{eq:unitary-duhamel-channel}
\end{equation}
\end{lemma}
\begin{proof}
For every reference system and state $\rho$, the channel difference is
$(U-V)\rho U^\dagger+V\rho(U^\dagger-V^\dagger)$.  Its trace norm is at most $2\|U-V\|_\infty$, uniformly in the reference dimension and input state.  Taking the diamond supremum proves Eq.~\eqref{eq:unitary-channel-perturbation}.  Duhamel's formula gives $\|U(T)-V(T)\|_\infty\le\int_0^T\|H(t)-\widetilde H(t)\|_\infty dt$, which proves Eq.~\eqref{eq:unitary-duhamel-channel}.
\end{proof}

\begin{proposition}[Stability under a transparency tolerance]
\label{prop:transparency-tolerance}
For the same reference processes, task, supports, and budget caps, relax
only the fine window-pair condition to
\begin{equation}
 \inf_{\theta}\|V_2V_1-e^{i\theta}I\|_\infty\le\delta
 \label{eq:transparency-tolerance}
\end{equation}
where $\theta$ ranges over real phases, on each control trajectory,
keeping paired boundary controls and all other
restrictions.  Denote the resulting performance optima by $E_f^\delta$.
Then
\begin{equation}
 E_f^\delta(A)=0,\qquad
 [p-\delta]_+\le E_f^\delta(B)\le p.
 \label{eq:tolerance-gap}
\end{equation}
For the common window perturbation considered below, the same enlarged
class satisfies
\begin{equation}
 E_f^\delta(B_{\rm actual})-E_f^\delta(A_{\rm actual})
 \ge p-\delta-4|\eta|w.
 \label{eq:tolerance-robust-gap}
\end{equation}
Thus the strict separation persists for $\delta+4|\eta|w<p$.
These are performance bounds for a relaxed control class; the resource
conversions defined above retain exact ideal transparency.
\end{proposition}
\begin{proof}
For $B$, write $C=V_2V_1$ in
Eq.~\eqref{eq:strict-B-general-endpoints}.
Up to a scalar phase its joint propagator is within operator norm $\delta$
of the paired comparator with $C=I$.
Lemma~\ref{lem:unitary-channel-perturbation}, followed by common boundary
conjugation and environment trace, bounds the channel discrepancy by
$\delta$ in half-diamond distance.
For mixtures the same bound holds against the mixture of their paired
comparators, whose error is $p$ by the common Bell witness.
The TV triangle inequality therefore gives $E_f^\delta(B)\ge[p-\delta]_+$.
The exact paired echo remains feasible for $A$, while a paired control
remains feasible for $B$, giving the other bounds.
The common window perturbation has integrated norm at most $2|\eta|w$
uniformly over this enlarged control class.
Applying the same discrepancy bound to both processes gives
Eq.~\eqref{eq:tolerance-robust-gap}.
\end{proof}

As a change of the declared physical process, add the same coupling to both control windows and write its Hamiltonian as $H_{\mathrm{ref}}+\Delta H$.  If the allowed class obeys the uniform bound
\begin{equation}
 \int_0^T\|\Delta H(t)\|_\infty dt\le r,
 \label{eq:strict-perturbation}
\end{equation}
then Lemma~\ref{lem:unitary-channel-perturbation} and trace-norm contractivity imply, uniformly over $u\in\mathfrak C_{\mathrm{win}}$,
\begin{equation}
 \frac12\left\|\mathcal E_{\mathrm{actual}}^u-\mathcal E_{\mathrm{ref}}^u\right\|_\diamond\le r,
 \qquad
 |E_{\mathrm{actual}}^*-E_{\mathrm{ref}}^*|\le r.
 \label{eq:strict-stability}
\end{equation}
For a common window coupling $\eta Z_S\otimes Z_E$ in both processes, all added terms commute with the two $ZZ$ coupling segments.  Hence the two actual uncontrolled endpoint channels remain exactly identical, with total phase $2gh+2\eta w$.  One has $r_A=r_B\le2|\eta|w$, and therefore
\begin{equation}
 E_B^*-E_A^*\ge\sin^2(2gh)-4|\eta|w.
 \label{eq:strict-robust-gap}
\end{equation}
The modified cPTs have identical coarse channels because the added interaction is common and commutes with the original coupling.  More explicitly, both actual uncontrolled channels are $\mathcal D_{2gh+2\eta w}$, and the compensated coarse optimum is exactly $\sin^2(2gh+2\eta w)$ by the same Bell witness.  For $g=0.5$, $h=1$, $\Omega=20$, and $|\eta|=0.1$, conservative rounding gives $E_A^*\le0.03142$, $E_B^*\ge0.67665$, and a certified gap of at least $0.6452$.  A sufficient strict-separation condition is $\sin^2(2gh)>4|\eta|w$; with $w=\pi/\Omega$ it becomes $|\eta|<\Omega\sin^2(2gh)/(4\pi)$.  For a perturbation of another operator form, the uniform stability bound remains valid, but exact equality of the two actual coarse channels would require a separate commutation and equal-phase check.

The perturbation estimate also holds uniformly for boundary pairs and their random mixtures.  Thus the same process $A$ retains a strict fine-versus-coarse advantage,
\begin{equation}
 E_c(GA_{\mathrm{actual}})-E_f(A_{\mathrm{actual}})
 \ge p-4|\eta|w.
 \label{eq:robust-access-gap}
\end{equation}
In the commuting perturbation model above, the coarse error for $A$ is known exactly.  Combining this value with $E_f(A_{\mathrm{actual}})\le2|\eta|w$ also gives the sharper, phase-resolved certificate
\begin{equation}
 E_c(GA_{\mathrm{actual}})-E_f(A_{\mathrm{actual}})
 \ge \sin^2(2gh+2\eta w)-2|\eta|w,
 \label{eq:phase-resolved-access-gap}
\end{equation}
whenever the right-hand side is positive.  Equation~\eqref{eq:robust-access-gap} remains the uniform reference-based bound used for the two-process comparison.
The proof compares both classes with their reference optima in Theorem~\ref{thm:strict-access}.  The two actual uncontrolled channels are still exactly equal; only the temporal allocation of the interaction distinguishes their control opportunities.  Consequently, Proposition~\ref{prop:endpoint-prediction} also gives the robust endpoint-only prediction lower bound $[p-4|\eta|w]_+/2$.

\subsection{Control-induced dissipative discrepancies}
\label{sec:dissipative-stability}

The cPT itself already contains the environmental action, including any
environment-induced dissipation and memory that belongs to the declared
process.  The present extension concerns a different comparison: a physical
control or instrument implementation is compared with a certified reference
implementation of that same cPT.  Thus $\Delta\mathcal L$ below denotes
control-induced damping, leakage, calibration error, or an analogous hardware
discrepancy.  An environmental dissipator that is part of the modeled cPT must
instead be included in $\mathcal L^{X,0}$ and in the process object itself.
The norm must be evaluated on the full retained system--environment--instrument
space; an endpoint system channel is not a replacement for an internal block.
Let $\mathcal L^{X,0}_{u,t}$ be the reference time-local generator for
process $X\in\{A,B\}$ and control branch $u$, and let
\begin{equation}
 \mathcal L^{X}_{u,t}=\mathcal L^{X,0}_{u,t}+\Delta\mathcal L^{X}_{u,t}
 \label{eq:dissipative-generator-perturbation}
\end{equation}
generate time-local CPTP propagators on that retained space, including every memory and record register crossing a block boundary. The Duhamel argument uses contractivity of each forward propagator. Both generator families have integrable bounded differences on the common retained space. Leakage is represented by additional retained levels with a common target and readout, so the full channels remain trace preserving. Define
\begin{equation}
 r_X:=\sup_{u\in\mathfrak C_{a_f\to\tau}(B)}
 \frac12\int_0^T
 \|\Delta\mathcal L^{X}_{u,t}\|_\diamond\,dt .
 \label{eq:dissipative-perturbation-radius}
\end{equation}
The supremum includes paid boundary controls and their admitted mixtures, with a bound uniform on each trajectory. For the displayed window integral, $T=2h+2w$ and boundary implementations are exact; if their implementation also varies, its certified discrepancy is added to $r_X$. The comparison uses the same physical task and record interface.

\begin{proposition}[Control-induced dissipative stability of the temporal gap]
\label{prop:dissipative-stability}
Under Eqs.~\eqref{eq:dissipative-generator-perturbation}--\eqref{eq:dissipative-perturbation-radius}, assume the bound holds uniformly for every nominal trajectory and the paired implemented controls have the same task interface and budget. The optimized fine errors of these paired realized controls obey
\begin{equation}
 \left|E_X^{\rm impl}-E_X^*\right|\le r_X,
 \qquad X\in\{A,B\},
 \label{eq:dissipative-optimality-stability}
\end{equation}
Here $E_X^{\rm impl}$ is the infimum over the implemented controls paired with the declared window class, and $E_A^*=0$ and $E_B^*=p$ are the reference values of
Theorem~\ref{thm:strict-access}.  Consequently,
\begin{equation}
 \boxed{
 E_B^{\rm impl}-E_A^{\rm impl}
 \ge p-r_A-r_B .}
 \label{eq:dissipative-gap}
\end{equation}
For complete finite-duration records with internal interventions, a uniform half-strategy distance between the corresponding entire open blocks gives the same conclusion; the generator integral in Eq.~\eqref{eq:dissipative-perturbation-radius} applies directly only when controls are inserted between complete blocks or when the retained joint generator includes those interventions.  General dissipative
perturbations need not preserve the equality of the two uncontrolled coarse
channels; that equality requires an additional commutation or covariance
condition.
\end{proposition}
\begin{proof}
For each branch, Duhamel's formula gives
\begin{equation*}
 \Phi^X_{u,T,0}-\Phi^{X,0}_{u,T,0}
 =\int_0^T
 \Phi^X_{u,T,t}\circ\Delta\mathcal L^X_{u,t}\circ
 \Phi^{X,0}_{u,t,0}\,dt .
\end{equation*}
The CPTP propagators on the two sides have diamond norm one, so the half
diamond distance of the complete propagators is at most $r_X$.  Contractivity
of every declared task record and the triangle inequality give the same bound
for each task error.  Taking the infimum over the common control class gives
Eq.~\eqref{eq:dissipative-optimality-stability}; applying it to the reference
values $0$ and $p$ yields Eq.~\eqref{eq:dissipative-gap}.
\end{proof}

For a control-induced Lindblad discrepancy with fixed jump operators,
\begin{equation}
 \Delta\mathcal L_t(\rho)
 =-i[\Delta H_t,\rho]
 +\sum_j\delta\gamma_j(t)\left(
 L_j\rho L_j^\dagger-\frac12\{L_j^\dagger L_j,\rho\}\right),
 \label{eq:lindblad-perturbation}
\end{equation}
one has
\begin{equation}
 \frac12\|\Delta\mathcal L_t\|_\diamond
 \le \|\Delta H_t\|_\infty
 +\sum_j|\delta\gamma_j(t)|\,\|L_j\|_\infty^2,
 \label{eq:lindblad-radius-bound}
\end{equation}
because the jump map has full diamond norm at most $\|L_j\|_\infty^2$,
and the two anticommutator maps together have full diamond norm at most
$\|L_j\|_\infty^2$.  Thus each Lindblad dissipator has full diamond norm at
most $2\|L_j\|_\infty^2$, while the commutator has full diamond norm at most
$2\|\Delta H_t\|_\infty$.  Dividing the full generator bound by two gives
Eq.~\eqref{eq:lindblad-radius-bound}.  Both total generators must retain
CPTP forward propagators; the rate differences themselves may have either
sign. Thus, for one
unit-norm control-induced dissipative channel of rate $\gamma$ acting in the
two windows,
$r_A,r_B\le2\gamma w$ and
\begin{equation}
 E_B^{\rm impl}-E_A^{\rm impl}\ge p-4\gamma w .
 \label{eq:lindblad-gap-bound}
\end{equation}

\subsection{An exact covariant dissipative example}

The covariant isotropic control-error model can be solved more sharply. Let the same control-layer depolarizing rate profile act over the protected interval, independently of the chosen trajectory and on the same declared supports. Write
\begin{equation}
 \mathcal N_\alpha=\alpha\,\mathrm{id}+(1-\alpha)\mathcal R_{I/2},
 \qquad
 \alpha=e^{-\int_0^T\gamma(t)dt}.
 \label{eq:isotropic-channel-alpha}
\end{equation}
The two processes retain a common uncontrolled endpoint channel $\mathcal N_\alpha\circ\mathcal D_p$. Indeed $Z_E$ is conserved, the initial bath state is diagonal in its $Z_E$ basis, and system-only controls preserve the resulting two bath sectors. Within each sector the system depolarizer commutes with the conditional system-unitary evolution. It therefore factors as $\mathcal N_\alpha$ after the noiseless reduced channel, also for reference-assisted inputs. The echo gives
$\mathcal E_A=\mathcal N_\alpha$, whereas
$\mathcal E_B=\mathcal N_\alpha\circ\mathcal D_p$.  For every paired control, the noiseless reduced channel is unital and its ideal Bell overlap is at most one; for $B$ that overlap is exactly $1-p$ in every branch. The depolarizer changes this overlap to $\alpha F+(1-\alpha)/4$. The corresponding Bell deficits give lower bounds for all branches and their mixtures. The echo for $A$ and the Pauli channel for $B$ attain them, yielding
\begin{equation}
 \boxed{
 E_A^{\rm impl}=\frac34(1-\alpha),\qquad
 E_B^{\rm impl}=\frac34(1-\alpha)+\alpha p,\qquad
 E_B^{\rm impl}-E_A^{\rm impl}=\alpha p .}
 \label{eq:isotropic-dissipative-gap}
\end{equation}
The coarse paired-unitary optimum is the process-$B$ value, so the same-process fine-versus-coarse gap for $A$ is also $\alpha p$. Boundary controls are noiseless in this example, or contribute a common trajectory-independent depolarizing exposure included in $\alpha$.

Thus a covariant control error attenuates the temporal advantage by the
distinguishability factor $\alpha$.  At the complete-erasure limit $\alpha=0$, no strict gap remains; every finite integrated rate has $\alpha>0$.  If the same isotropic dissipator is instead part of the physical
environment defining the cPT, it is incorporated into $\Upsilon_A$ and
$\Upsilon_B$ rather than appended as an external perturbation; the displayed
calculation then describes that explicitly enlarged cPT model.  For a noncovariant or memory-bearing control error, the general bound \eqref{eq:dissipative-gap} applies when a joint time-local CPTP dilation supplies the stated generator bound, or when an entire-block strategy-distance certificate is supplied; no universal positive gap follows from CPTP contractivity
alone.

\subsection{Transporting implementation bounds to other tasks}

The paired-implementation lemma applies to every fixed task $\tau$. Under the test correspondence of Eq.~\eqref{eq:unitary-test-covariance}, the discrepancy of two output processes $\Lambda_1,\Lambda_2$ satisfies
\begin{equation}
 \sup_{\mathcal J_U\in\Inst_U}
 \TV(p_{\Lambda_1}^{\mathcal J_U},p_{\Lambda_2}^{\mathcal J_U})
 =
 \sup_{\mathcal J\in\Inst_{\rm id}}
 \TV(p_{\Gamma_U\Lambda_1}^{\mathcal J},p_{\Gamma_U\Lambda_2}^{\mathcal J}).
 \label{eq:unitary-implementation-discrepancy}
\end{equation}
This follows outcome by outcome from the terminal-frame map, with both processes transformed by the same frame. The coherent and dissipative estimates therefore apply to unitary target tasks with the same discrepancy constants when the paired controls, including paid boundary gates, satisfy the common budget. Applied to Proposition~\ref{prop:direct-unitary-qubit}, a uniform radius $r_X$ on each complete gate implementation gives $E_{\tau_U,f,B}^{\rm impl}(B)-E_{\tau_U,f,B}^{\rm impl}(A)\ge p-r_A-r_B$. The radius includes imperfections of the boundary gate as well as the interior pulses. For multi-time tasks, the same argument uses the declared complete record family and retained instrument networks.

\section{CPRC as a physical free conversion}
\label{sec:cprc}

\subsection{Regular dilations}

On the regular finite-dimensional dilation class, let the laboratory Hamiltonian be
\begin{equation}
 H_\omega(t)=\id_S\otimes H_E(t)+H_{SE}(t)+H_{\mathrm{ctrl},\omega}(t)
 +H_{\mathrm{inst},\omega}(t),
 \label{eq:lab-hamiltonian}
\end{equation}
where $U_\omega(t)$ is an absolutely continuous unitary path generated by an integrable bounded system control.  The system--environment interaction is bounded on the finite-dimensional dilation, and finite-duration instruments satisfy pointwise covariance,
\begin{equation}
 H_{\mathrm{inst},\omega}(t)=
 (U_\omega(t)\otimes\id_A)H_{\mathrm{inst}}^{\log}(t)
 (U_\omega^\dagger(t)\otimes\id_A).
 \label{eq:pointwise-covariance}
\end{equation}
For a local instrument event $dx$ occurring at time $r$, covariance acts on its complete CP event kernel as
\begin{equation}
 \mathcal M^{\rm lab}_{\omega,dx;r}
 =\mathcal U_{\omega,r}\circ\mathcal M^{\log}_{dx;r}
 \circ\mathcal U_{\omega,r}^{-1},
 \qquad
 \mathcal U_{\omega,r}:=\operatorname{Ad}_{U_\omega(r)\otimes I_A}.
 \label{eq:event-kernel-covariance}
\end{equation}
For the full input--output event kernel of a finite-duration block supported on $[s,t]$, the entrance and exit frames are distinct:
\begin{equation}
 \mathcal M^{\rm lab}_{\omega,dx;[s,t]}
 =\mathcal U_{\omega,t}\circ\mathcal M^{\log}_{\omega,dx;[s,t]}
 \circ\mathcal U_{\omega,s}^{-1}.
 \label{eq:finite-event-kernel-covariance}
\end{equation}
The frame acts trivially on every retained environment, reference, memory, and record factor other than $S$; identities on these factors are implicit.  Here $\mathcal M^{\log}_{\omega,dx;[s,t]}$ is the actual block expressed in logical coordinates.  With $H_{SE}=0$, it equals the prescribed logical block $\mathcal M^{\log}_{dx;[s,t]}$ for every path.  With $H_{SE}\ne0$, it retains the rotated interaction and can depend on $\omega$, as in Eq.~\eqref{eq:window-logical-generator}.  Thus Eq.~\eqref{eq:finite-event-kernel-covariance} does not identify a noisy block with the uncoupled reference.  A block with multiple open ports carries the frame at each port's own time; adjacent entrance and exit frames cancel when the ports are linked.  Adapted local feedback maps obey Eq.~\eqref{eq:event-kernel-covariance}, and the terminal readout includes the final-frame compensation.  Record labels are unchanged, or undergo one fixed process-independent measurable relabeling included in the task interface.  The input is encoded in the initial frame.  Pointwise covariance of the generator and local event kernels then preserves the full ideal record law in the absence of $H_{SE}$.

Let $\Proc$ denote these regular physical realizations.  A concrete measurable path model uses Hermitian generators $K\in L^1([0,T])$ and initial frames in $U(d)$, with their standard Borel structures.  The integral equation $\dot u=-iKu$ has a unique absolutely continuous solution; finite-dimensional propagators and CP instrument contractions are measurable in these data.  Encoding and readout are included in the path network.  More general parameter spaces must supply measurable event kernels as part of their instrument model.  For a fixed path $u$, denote the resulting physical instrument pullback by $\mathcal T_u^*$.  A probability law $\mu$ on this standard Borel implementation space defines
\begin{equation}
 (\Theta_\mu\Lambda)[\mathcal J']
 =\int\mu(du)\,\Lambda[\mathcal T_u^*\mathcal J'].
 \label{eq:path-average}
\end{equation}
The law is initially independent of the protected input.  More precisely, after the input has been encoded in the initial control frame, independence is imposed on the logical joint state and the Markov control state; the laboratory state may then be correlated with the frame label by this prescribed encoding.  The integrand is a measurable probability kernel for every instrument event.  The main task interface contains all logical instrument outcomes, with the internal control path averaged out.  If a task also retains a control label, its law and the enlarged ideal interface must be declared before optimization; a change of label law is then a change of that task specification.  Path integration defines a normalized instrument with an independent classical controller, rather than a postselected process.

Equation~\eqref{eq:path-average} agrees with the record-preserving Born rule of the Hamiltonian realization.  A dilation theorem can justify approximation of a continuum process by such realizations in the topology of its own statement \cite{fuksa2026}.  The general extension below instead uses physical instrument pullback directly; it does not infer any Fock-norm bound from density.  Ideal Brownian paths are treated in the stochastic formulation of Ref.~\cite{wang2026cprc}, while finite-budget Hamiltonian implementations here use regular paths or explicitly costed regular approximations.

\begin{proposition}[Pathwise transparency]
\label{prop:pathwise}
If $H_{SE}=0$, pointwise generator and local-event covariance, Eqs.~\eqref{eq:pointwise-covariance} and \eqref{eq:event-kernel-covariance}, hold on every instrument support, finite-duration blocks obey Eq.~\eqref{eq:finite-event-kernel-covariance}, and adapted feedback and readout are covariant, then every control trajectory has exactly the ideal task distribution.
\end{proposition}
\begin{proof}
The frame acts as $U_\omega$ on $S$ and as the identity on the environment, reference, instrument memories, and classical records.  Define $\widetilde\rho_\omega=(U_\omega^\dagger\otimes I)\rho_\omega(U_\omega\otimes I)$, with $I$ on all those other factors.  Absolute continuity permits differentiation almost everywhere, and the frame derivative cancels the control Hamiltonian.  Equation~\eqref{eq:pointwise-covariance} becomes the fixed logical instrument Hamiltonian.  Initial encoding removes the initial frame, and common environment propagation preserves this equality.  Each logical record event and its adapted feedback map is the prescribed one.  Equation~\eqref{eq:finite-event-kernel-covariance} uses the frames at the actual block endpoints; successive endpoint factors cancel along the network.  Terminal compensation then gives the ideal law on the complete record.
\end{proof}

\begin{proposition}[CPRC closure on the free set]
\label{prop:cprc-free}
Suppose each regular path network respects the entrance cut, satisfies pathwise transparency, and has a physical instrument pullback at the declared interfaces.  Its normalized input-independent randomization is a free conversion.  For two admissible randomized networks with independent path laws $\mu$ and $\nu$, their composition is
\begin{equation}
 \Theta_\nu\circ\Theta_\mu=\Theta_{\nu\star\mu},
 \qquad
 \Cost(\Theta_{\nu\star\mu})\le\Cost(\Theta_\mu)+\Cost(\Theta_\nu),
 \label{eq:cprc-composition}
\end{equation}
\end{proposition}
\begin{proof}
Proposition~\ref{prop:free-preserve} applies pathwise, and randomization preserves the entrance-erasing form.  Ideal transparency is pathwise transparency.  The composite pullback is $\mathcal T_u^*\mathcal T_v^*$, in this order.  Define $\nu\star\mu$ as the pushforward of $\mu\otimes\nu$ to the stacked physical network.  Fubini's theorem for bounded event probabilities proves Eq.~\eqref{eq:cprc-composition}.  With Eq.~\eqref{eq:pointwise-covariance}, the composite action on every retained apparatus or instrument block $G_{\rm app}$ is $\mathcal A_u(\mathcal A_v(G_{\rm app}))=\mathcal A_{uv}(G_{\rm app})$; it does not conjugate the fixed background $H_{SE}+I_S\otimes H_E$.  Whenever the full path action has this representation, the composite path is $u(t)v(t)$.  Its generator is $K_u+uK_vu^\dagger$, giving the norm cost bounds of Lemma~\ref{lem:window-closure}.  Otherwise the stacked-network definition is used.  No invariance of the path law under multiplication is required.  Composition stays in the generated architecture, although it need not stay in a fixed-refresh or diffusion parametrization.  Sampler and record costs are included as in Eq.~\eqref{eq:cost-axioms}.
\end{proof}

\subsection{Extension to general cPTs by experiment pullback}

For a general $\Lambda\in\CPT_a$, define the CPRC extension by its action on experiments,
\begin{equation}
 \boxed{
 (\Theta_U\Lambda)(\mathcal J'):=
 \Lambda(\Theta_U^*\mathcal J'),
 }
 \label{eq:cprc-pullback}
\end{equation}
Here $\Theta_U$ denotes a fixed admissible path law as well as its implementation.  The pullback $\Theta_U^*$ is the physical insertion of that implementation into an output experiment, including every record map and the independent random controller.  For the extension require
\begin{equation}
 \Theta_U^*\Inst_{a'}\subseteq\Inst_a.
 \label{eq:pullback-closure}
\end{equation}
In addition to this set inclusion, the pullback must be affine under instrument randomization and compatible with the multilinear slot structure.  It must commute with outcome summation and measurable outcome relabeling, send deterministic instruments to deterministic instruments, commute with ancillary extension, and preserve causal prefixes and their completion.  Its event kernels are measurable and countably additive, and it maps admissible instrument limits to admissible limits continuously.  The following proposition verifies these properties in a concrete regular model.  An ideal stochastic implementation is included only when its stochastic instrument model or established limit satisfies the same properties.

\begin{proposition}[Physical structure of a regular pullback]
\label{prop:regular-pullback}
Consider finite-dimensional physical insertion networks with fixed causal supports, CP instrument kernels, initially independent random controls on the measurable path model above, and normalized record-preserving completion.  If controls are inserted only between complete instrument blocks, equip finite instruments with the sum of their complete-block diamond distances.  If a control is inserted inside a finite-duration block, represent that block as an open comb or strategy and equip it with the corresponding complete strategy norm; alternatively, retain a joint dilation on a common system--environment--memory--record space and use the $L^1$ norm of its full joint-generator difference on the retained support.  In this generator topology, auxiliary initial states, non-Hamiltonian event kernels, adapted feedback maps, and terminal record-extraction channels are fixed.  If these data vary, include their trace-distance or complete record-preserving diamond discrepancies in the topology as well, uniformly over past records for adapted maps.  The record spaces and any relabeling are common to the comparison.  An endpoint-channel diamond norm alone does not supply continuity for an internal insertion.  If the chosen source hardware admits the inserted controls, their instrument pullback has the algebraic, ancillary, prefix, measure, and continuity properties stated above.  For the completion version, fix compatible source and target instrument completions $\overline{\Inst_{a'}}$ and $\overline{\Inst_a}$, with finite networks dense in the domain and $\overline{\Inst_a}$ complete in the selected target topology.  Assume outcome maps, ancillary extensions, and prefix operations extend continuously on these spaces, and limits of pulled-back networks remain normalized, completely positive, countably additive, record-preserving, causally admissible instruments in the declared target hardware class.  Then the pullback extends uniquely as
\begin{equation}
 \Theta_U^*: \overline{\Inst_{a'}}\longrightarrow\overline{\Inst_a}
 \label{eq:completed-pullback-map}
\end{equation}
with the same structural properties on these completions.
For randomized generator realizations with one fixed path law $\mu$,
retain the same independent classical controller internally in the target
instrument and use the topology
\[
 \operatorname*{ess\,sup}_{u\sim\mu}
 \int\|G_{{\rm app},u}(t)-\widetilde G_{{\rm app},u}(t)\|_\infty\,dt,
\]
together with the stated discrepancies for any varying event data.
This is a uniform topology on measurable path families; the target
completion retains their common controller law.  Averaging out the
controller at the task output is the same fixed channel on both families.
In an open-strategy formulation all admissible insertion ports are
retained, so equality of strategies implies equality after insertion.
An equivalence that discards those ports is not used for this pullback.
\end{proposition}
\begin{proof}
At a fixed path, linking CP instrument kernels and control channels gives a CP kernel, multilinear in distinct slots and affine in classical mixing.  Outcome summation gives the linked deterministic CPTP network, while relabeling pushes forward the same kernel.  Tensoring an untouched reference commutes with linking.  Causal prefix restriction deletes later completed blocks and so commutes with the pullback.  Finite compositions of measurable kernels are measurable and countably additive.  Integration against the independent normalized path law preserves these statements by Tonelli's theorem and preserves normalization without postselection.  For continuity, when the control is inserted only between complete instrument blocks, replacing the blocks one at a time bounds the output trace-norm discrepancy by the sum of their complete-block diamond discrepancies, uniformly in the inserted unitary path.  This argument uses the complete-block topology only in this boundary-insertion case.  If the control acts inside a finite-duration block, the block is retained as an open comb or strategy and its discrepancy is measured in the corresponding complete strategy norm; an endpoint channel norm alone does not control such an internal insertion.  For bounded Hamiltonian instruments with the internal support retained, Duhamel's estimate gives
\begin{equation}
 \frac12\|\mathcal U_H-\mathcal U_{\widetilde H}\|_\diamond
 \le \int_0^T\|H(t)-\widetilde H(t)\|_\infty dt,
 \label{eq:internal-duhamel}
\end{equation}
uniformly over the inserted control path, where $\mathcal U_H=\operatorname{Ad}_{U_H(T)}$ and $\mathcal U_{\widetilde H}=\operatorname{Ad}_{U_{\widetilde H}(T)}$ are the full joint unitary channels.  For a retained dilation with fixed auxiliary preparation and event/readout maps, applying these common channels cannot increase the induced operational distance.  More explicitly, a block $\Phi_{H,\eta,\mathcal Q}(\rho)=\mathcal Q[U_H(\rho\otimes\eta)U_H^\dagger]$, with a complete CPTP record-extraction channel $\mathcal Q$, satisfies
\begin{align}
 \frac12\|\Phi_{H,\eta,\mathcal Q}
 -\Phi_{\widetilde H,\widetilde\eta,\widetilde{\mathcal Q}}\|_\diamond
 &\le \int_0^T\|H(t)-\widetilde H(t)\|_\infty\,dt\notag\\
 &\quad+\frac12\|\eta-\widetilde\eta\|_1
 +\frac12\|\mathcal Q-\widetilde{\mathcal Q}\|_\diamond.
 \label{eq:instrument-data-continuity}
\end{align}
This follows by changing the generator, auxiliary state, and readout one at a time, using Lemma~\ref{lem:unitary-channel-perturbation} and complete-channel contraction.  The last two terms vanish for fixed preparation and readout; variations of additional event or feedback blocks are bounded by the same telescoping argument on their complete record-preserving channels.  Telescoping over finitely many blocks gives the sum of their complete-block discrepancies in the boundary-insertion case, while the open-strategy or joint-generator estimate gives the corresponding bound for an internal insertion.  Hence the pullback is uniformly continuous in the selected topology.  Let $\{\mathcal J'_n\}\subset\Inst_{a'}$ be finite networks approximating an element of $\overline{\Inst_{a'}}$.  The uniform bound makes $\{\Theta_U^*\mathcal J'_n\}$ Cauchy in the complete target space $\overline{\Inst_a}$; define the value of Eq.~\eqref{eq:completed-pullback-map} as its limit.  Equivalent approximating sequences have equivalent images, so this definition is independent of the approximation and is the unique continuous extension.  The assumed target admissibility preserves normalization, complete positivity, countable additivity, and the causal record structure.  Continuity of the declared outcome, ancillary, and prefix operations passes their compatibility identities to the limit.  Thus the completion statement applies to the specified compatible spaces; the norm bound alone does not establish admissibility of a limit in an arbitrary hardware class.
For the frame pullbacks used in this paper there is also a direct estimate in the topology of the retained block.  With $\widehat F=F\otimes I$ on all retained environment, memory, and record factors,
\begin{equation}
 \bigl\|\mathcal A_F(G_{\rm app})-\mathcal A_F(\widetilde G_{\rm app})\bigr\|_{L^1}
 =\int_0^T\!dt\,\bigl\|\widehat F(t)(G_{\rm app}(t)-\widetilde G_{\rm app}(t))\widehat F(t)^\dagger\bigr\|_\infty
 =\|G_{\rm app}-\widetilde G_{\rm app}\|_{L^1}.
 \label{eq:frame-L1-isometry}
\end{equation}
Thus the background-independent apparatus action is continuous in the joint-generator topology.  If a control is inserted inside a block, the corresponding open-strategy norm must retain every open port, instrument-memory register, and record register on which that block can act; a terminal channel norm after these systems have been discarded is not an admissible substitute.  Equation~\eqref{eq:frame-L1-isometry}, or the complete strategy norm when the block is represented operationally, is the estimate used for the stated completion.
\end{proof}

\begin{theorem}[CPRC extension and operational contraction]
\label{thm:cprc-extension}
Assume Eq.~\eqref{eq:pullback-closure} and the physical, algebraic, measure, prefix, ancillary, and continuity properties stated above.  Then Eq.~\eqref{eq:cprc-pullback} defines a unique cPT operational equivalence class on $a'$, agrees with the regular network construction, and satisfies, for all $\Lambda,\Lambda'\in\CPT_a$,
\begin{equation}
 d_{a'}(\Theta_U\Lambda,\Theta_U\Lambda')
 \le d_a(\Lambda,\Lambda').
 \label{eq:tv-contraction}
\end{equation}
Equality holds if the conversion has an inverse physical pullback on the declared experiment classes that is mutually inverse at the level of all operational event laws.
\end{theorem}
\begin{proof}
Complete positivity follows by applying the pulled-back experiment at each ancillary extension, where it is an admissible physical instrument.  Affineness and slot compatibility give the Born functional structure.  Summing outcomes pulls back to a deterministic experiment, so total mass is one.  Countable additivity follows from the compatible event maps and probability kernels.  For causality, restriction to an output prefix commutes with pullback; changing a later output instrument therefore cannot change the pulled-back prefix law.  Continuity of both the instrument pullback and the input Born functional gives continuity on admissible limits.  This establishes the full operational cPT definition.  The event laws uniquely specify the output operational class, and on regular realizations they coincide with Eq.~\eqref{eq:path-average}.  For the distance,
\begin{align}
 d_{a'}(\Theta_U\Lambda,\Theta_U\Lambda')
 &=\sup_{\mathcal J'\in\Inst_{a'}}
 \TV\!\left(p_\Lambda^{\Theta_U^*\mathcal J'},
 p_{\Lambda'}^{\Theta_U^*\mathcal J'}\right)\\
 &\le d_a(\Lambda,\Lambda'),
 \end{align}
because the pullback set is a subset of $\Inst_a$.  A mutually inverse physical pullback gives the reverse inequality by the same argument.  A bijection of redundant experiment labels is not necessary; operational event equivalence suffices.
\end{proof}

\begin{remark}[Operational and representation-dependent continuity]
The cPT Fock representation carries an infinite direct sum of particle sectors.  Density of a regular core does not imply boundedness of a conversion in an unweighted norm.  Operational continuity above is sufficient for the resource theory.  A norm estimate in a chosen Fock completion additionally requires the conversion to preserve that completion with a uniform analytic bound.
\end{remark}

\begin{remark}[Extension in another completion]
\label{prop:weighted-gamma}
If a specified weighted norm $\|\cdot\|_{\Gamma_\rho}$ has a dense linear regular core and a linear conversion satisfies
\begin{equation}
 \|\Theta_U X\|_{\Gamma_\rho}\le M_U\|X\|_{\Gamma_\rho},
 \label{eq:weighted-gamma-bound}
\end{equation}
then it extends uniquely to that norm completion: images of core Cauchy sequences are Cauchy, equivalent sequences have equivalent images, and the estimate passes to the limit.  This elementary conditional extension is not a graph estimate proved for CPRC.  No choice of weight or value of $M_U$ is needed by the operational results in this paper.
\end{remark}

\section{CPRC complete-record performance certificates}
\label{sec:bound}

Two task specifications occur in this paper.  The finite-window separation in Sec.~\ref{sec:strict-model} uses the endpoint identity transmission task with no intermediate logical intervention.  This section uses the uncoupled reference-process task with complete logical records and possibly adaptive finite-duration instruments.  Their errors and control classes are kept distinct.  We use the complete-record performance theorem of Ref.~\cite{wang2026cprc} as the dynamical input, and derive from it an achievable task error within the declared control architecture and budget.  The complete-record bound supplies a constructive capability certificate once its physical hypotheses and implementation budget are satisfied.

\subsection{CPRC performance guarantee}

For this application, the prescribed task $\Xi_\tau$ is the uncoupled reference process: it retains the initial joint state, the time-ordered environment propagation, the logical instruments, and the readout, and removes the centered system--environment interaction.  The complete-experiment theorem in Ref.~\cite{wang2026cprc} applies under the following conditions:
\begin{enumerate}
\item The full Markov control state $Z_t$ is initially independent of the joint quantum state in logical coordinates, after the prescribed initial-frame compensation.  Its transition law is autonomous and is not changed by feedback from the system or measurement records.
\item The marginal laws $\pi_t$ and transition kernels $P_{s,t}$ satisfy $\pi_sP_{s,t}=\pi_t$.  For OU control the full state includes both the frame and the driving variables.
\item The rotated interaction is bounded and centered at almost every time,
\begin{equation}
 \|H_{\mathrm{int},z}(t)\|_\infty\le g(t),
 \qquad
 \int\pi_t(dz)H_{\mathrm{int},z}(t)=0,
 \quad g\in L^1[0,T],
 \label{eq:bounded-centered}
\end{equation}
where $g$ is nonnegative and independent of the logical experiment in the comparison family.
\item Every finite-duration instrument has pointwise-covariant generator and local event kernels, and its full block obeys the two-endpoint relation in Eq.~\eqref{eq:finite-event-kernel-covariance}.  In logical coordinates, the prescribed local event maps, adapted feedback rule, record-label map, and terminal readout are independent of $z$ after frame compensation; a noisy propagated block still includes the rotated interaction.  The compared output is the complete logical measurement record after frame compensation and averaging.  The full control trajectory is not appended to this output; a task giving the tester that additional side information requires a separate bound.
\end{enumerate}
The control can become statistically correlated with the records it generates; autonomy specifies its transition law rather than statistical independence at later times.  The result used is Theorem~2, ``Complete-experiment operational error bound,'' in the subsection ``Process convergence and complete-record error'' of Ref.~\cite{wang2026cprc}; its derivation is in Appendix~G, subsection ``Nonautonomous complete-record theorem.''  The statement and assumptions quoted here are those of Ref.~\cite{wang2026cprc}; in particular, the control transition law is autonomous, while correlations generated with later records are treated within the complete-record experiment.  This cited theorem is the external dynamical input for the CPRC bounds below, rather than a result reproved here.  The theorem and subsection titles specify the result if numbering changes in a later version.  Its complete experiments include reference-assisted initial states, adapted logical instrument maps, and terminal readout, with the control path averaged out.  Feedback may choose later logical instruments but cannot change the autonomous control transition law.

Centering is a dynamical hypothesis, not a consequence of dilation.  For a full unitary twirl on $S$, the averaged interaction is $I_S\otimes\Tr_S(H_{SE})/d$.  This bath-only term can be assigned to the common $H_E(t)$ before defining the interaction; the remaining system-traceless coupling is centered.  A restricted frame group may leave further invariant components.  Those components require a separate error estimate or an ideal task that retains them; without centering, the quoted fast-refresh guarantee is not available for their removal.

Define the average mixing envelope
\begin{equation}
 b(s,t):=\int\pi_s(dz)\,
 \TV\!\left(P_{s,t}(z,\cdot),\pi_t\right).
 \label{eq:mixing-envelope}
\end{equation}
The complete-experiment operational error theorem of Ref.~\cite{wang2026cprc} then gives
\begin{equation}
 \eps_\tau(\overline\Lambda_{\mathrm{CPRC}})
 \le \min\left\{1,\;4\int_0^Tdt\int_0^t ds\,
 g(t)g(s)b(s,t)\right\}.
 \label{eq:complete-record-bound}
\end{equation}
In the cited derivation, the interaction commutator has trace-norm bound
$2g(t)$ and the operator-valued control transfer has bound $2b(s,t)$.
Together they bound the control-dependent conditional component by
$4\int_0^t g(s)b(s,t)ds$.
A second commutator gives the joint-state discrepancy coefficient $8$;
the final record TV is at most half that trace-norm discrepancy, giving
the coefficient $4$ above.  Complete record-preserving instruments are
contractive on the retained joint state.  This summarizes the constant
accounting of the external theorem; its conditional evolution proof
remains in the cited source.
For application, choose $\Inst_\tau$ as a subfamily of those complete experiments, with a common $g$ and the preceding assumptions holding uniformly.  Its inclusion is explicit: an endpoint transmission tester has an arbitrary state on $QR$, no intermediate instrument, and one terminal POVM; it is the zero-intervention case of the cited family.  For a multi-time task, every declared logical instrument and its adapted rule must satisfy the same pointwise covariance on its whole support.  Append its reference and all logical records to the joint experiment used in the cited theorem.  If the task reports a measurable coarse-graining $f$ of that complete record, TV contracts under $f$; taking the supremum over this included task family gives Eq.~\eqref{eq:complete-record-bound}.  Unrestricted instrument completions are not covered by this argument unless the cited estimate is known to pass to that completion.  The conditional-state evolution, operator-valued transfer estimate, and complete-record proof remain in the cited reference.  An integrable covariance of an unbounded bath field alone does not supply its operator-norm premise.

\subsection{Protocol specializations}

For stationary control, let $b(s,t)=b(t-s)$ and $t_{\mathrm{mix}}=\int_0^\infty b(u)du<\infty$.  With the additional assumption $g\in L^2[0,T]$, the stationary specialization is
\begin{equation}
 \eps_\tau(\overline\Lambda_{\mathrm{CPRC}})
 \le\min\{1,4t_{\mathrm{mix}}\|g\|_2^2\}.
 \label{eq:stationary-bound}
\end{equation}
To see the specialization, extend $g$ by zero outside $[0,T]$ and set $b_+(u)=b(u)$ for $u\ge0$, zero otherwise.  The time triangle equals $\int g(t)(b_+*g)(t)dt\le\|g\|_2\|b_+*g\|_2\le t_{\mathrm{mix}}\|g\|_2^2$ by Cauchy--Schwarz and Young's inequality.  Scheduled refresh below is nonstationary and is evaluated directly, not by assigning it a stationary mixing time.
For independent uniform redraws of the complete $N_{\rm fr}$-frame control state on holding intervals $I_k$ of length $\Delta_k$, the mixing envelope is $1-N_{\rm fr}^{-1}$ within a holding interval and zero across a redraw.  Provided the group average centers the interaction, the fixed-refresh specialization gives
\begin{equation}
 \eps_\tau(\overline\Lambda_{\mathrm{CPRC}})
 \le\min\left\{1,2(1-N_{\rm fr}^{-1})g_*^2\sum_k\Delta_k^2\right\}
 \le\min\{1,2(1-N_{\rm fr}^{-1})g_*^2Th_*\},
 \label{eq:refresh-bound}
\end{equation}
where $g_*:=\operatorname*{ess\,sup}_{t\in[0,T]}g(t)<\infty$ and $h_*:=\max_k\Delta_k$.  Each holding interval contributes a triangle of area $\Delta_k^2/2$, which gives the first bound; $\sum_k\Delta_k^2\le Th_*$ gives the second.  Four Pauli frames yield $\eps_\tau\le\min\{1,(3/2)g_*^2Th_*\}$.  Finite Cayley transitions, Brownian control, and OU control use their own mixing envelopes and conditions from Ref.~\cite{wang2026cprc}.  A finite transition is not assigned the constant-frame envelope by identifying its nominal update interval.  In particular, a finite-amplitude realization must use its actual transition law and any separately bounded implementation discrepancy.

\subsection{Protocol costs and fair comparisons}
\label{sec:protocol-costs}

We use implementation costs that are demonstrably subadditive.  For an absolutely continuous frame with physical generator $K_P$, define, for $1\le q<\infty$,
\begin{equation}
 A_q(P):=\left(\int\|K_P(t)\|_\infty^qdt\right)^{1/q},
 \quad A_\infty(P):=\operatorname*{ess\,sup}_t\|K_P(t)\|_\infty.
 \label{eq:amplitude-costs}
\end{equation}
The integration includes all paid boundary and transition supports.  One concrete vector budget is
\begin{equation}
 \Cost(P)=\bigl(N_{\mathrm{upd}},A_1,A_q,A_\infty,
 \log_2 d_A,m_{\mathrm{rec}},n_{\mathrm{rand}}\bigr),
 \label{eq:protocol-cost}
\end{equation}
where $d_A$ is the allocated controller ancillary dimension, $m_{\rm rec}$ is allocated classical memory in bits, and $n_{\rm rand}$ counts supplied random bits in a finite sampling implementation.  Budgets are compared componentwise.  Temporal supports and any hardware bandwidth condition are part of the declared architecture; a hard bandwidth limit is imposed only when closure of that hardware class has been verified.

The two-component budget $B=(A_\infty,A_1)$ used in the endpoint double-window example is the corresponding projection of this cost vector for the system frame, with all other components held fixed.  In that example $G_{\rm app}=0$, so the physical control generator is $K_P=i\dot F F^\dagger$.  For a general instrument with a nonzero logical apparatus field $G_{\rm app}$, the laboratory field also contains $\widehat F G_{\rm app}\widehat F^\dagger$; its instrument and memory costs are then counted by the full vector $\Cost(P)$ rather than by the two-component endpoint projection.  This keeps the endpoint optimum and the general priced conversion theory on the same cost convention without identifying their different hardware components.

For a general task, this cost is a property of one shared implementation $P$ and is required to satisfy the declared budget simultaneously for every $\mathcal J\in\Inst_\tau$; it is not reoptimized separately for each tester.  If the logical instruments are supplied by the interface, only the incremental control, compilation, memory, and record hardware is charged.  If instrument implementation is part of the conversion, the corresponding worst-branch and worst-task costs are included in $\Cost(P)$.

For a composed regular frame $U_2U_1$, its generator is $K_2+U_2K_1U_2^\dagger$.  Operator-norm invariance and Minkowski's inequality give $A_q(P_2P_1)\le A_q(P_1)+A_q(P_2)$, including $q=1,\infty$.  Stacking hardware adds update counts, allocated memory, and random-bit supplies.  Independent ancillary dimensions multiply, so $\log_2(d_Ad_B)=\log_2d_A+\log_2d_B$.  Consequently Eq.~\eqref{eq:cost-axioms} holds for frame-representable physical compositions and for any larger declared hardware class whose implementation cost is defined by the same stacking rule.  A parametrized protocol family that is not closed under this physical stacking rule does not by itself define the conversion class.  A nonnegative weighted sum of these components is a scalar alternative.  A squared $L^2$ energy is a valid hardware constraint, but its root, rather than its square, obeys the stated additive-budget rule for simultaneous frame composition.

A fixed-angle pulse compressed to duration $w$ has $A_1$ independent of $w$, $A_q\propto w^{1/q-1}$ for $q>1$, and $A_\infty\propto w^{-1}$.  Thus shortening a pulse increases peak and higher-norm costs while leaving linear action fixed.  More redraws increase update and sampling costs.  Ideal Brownian frames are almost surely not absolutely continuous and have no ordinary generator cost as in Eq.~\eqref{eq:amplitude-costs}.  A finite-budget implementation uses a specified regular approximation with actual costs and approximation error.

A nondegenerate Gaussian OU generator has continuous sample paths and finite $A_q$ almost surely on a compact interval, but their law has unbounded support.  In particular, $\operatorname*{ess\,sup}_\omega A_\infty(\omega)=\infty$ for a nontrivial Gaussian drive.  It therefore does not satisfy a finite worst-branch cap.  A finite-budget architecture can instead admit a bounded drive or a capped realization, with its sampler, boundary compensation, actual mixing law, and implementation discrepancy counted.  An expected-cost architecture is another explicit choice.  A Poisson update count likewise requires a cap or an expected-cost convention.  Exact continuous sampling laws are ideal models; finite random-bit implementations use their declared discretization error or an explicitly specified analogue controller supply.

For example, couple an ideal regular generator $K_\omega$ to a capped implementation that agrees with it whenever $\sup_t\|K_\omega(t)\|_\infty\le L$.  If both networks have the same output records and logical compensation, and $q_L$ is the probability of the complementary event, then uniformly over common testers
\begin{equation}
 d_{a_\tau}(P_L\Lambda,P\Lambda)\le q_L,
 \qquad \eps_\tau(P_L\Lambda)\le\min\{1,u_P+q_L\}
 \label{eq:capped-tail-certificate}
\end{equation}
whenever $\eps_\tau(P\Lambda)\le u_P$.  Indeed the coupled record laws coincide on the good event, and their TV on the complement is at most one.  A pointwise saturation of a regular generator has $A_\infty\le L$ and interior $A_q\le LT^{1/q}$, but the compensated boundary supports and sampling hardware must also fit the budget.  Equation~\eqref{eq:capped-tail-certificate} uses an implementation comparison; direct use of the dynamical theorem for a modified drive additionally requires its own centering and mixing envelope.

For example, if $N$ independent holding intervals are allowed and each costs $c_{\rm draw}>0$, the draw budget requires $Nc_{\rm draw}\le B_{\rm draw}$.  Here the $N$ draws include the initial sample and $N-1$ subsequent redraws.  Among these ideal holding schedules, $\sum_k\Delta_k^2\ge T^2/N$, with equality at uniform spacing.  The best certificate from Eq.~\eqref{eq:refresh-bound} at that $N$ is therefore
\begin{equation}
 \eps_\tau\le\min\{1,2(1-N_{\rm fr}^{-1})g_*^2T^2/N\}.
 \label{eq:refresh-budget-certificate}
\end{equation}
For the draw constraint alone, this certificate is smallest at
$N_{\max}=\lfloor B_{\rm draw}/c_{\rm draw}\rfloor\ge1$.
This minimizes that upper-bound expression, not the true error over all controls.  Peak and transition constraints can further restrict feasible $N$ and alter the envelope.  For a fixed physical interaction controlled only by system-frame conjugation, $\|U^\dagger H_{SE}(t)U\|_\infty=\|H_{SE}(t)\|_\infty$, so one coupling envelope can be used across feasible schedules.  If the hardware control also changes the physical interaction, the envelope must instead be recomputed for that implementation.

\begin{lemma}[Zero-intervention bridge for the transmission task]
\label{lem:cprc-transmission-bridge}
Let $\Inst_\tau^{\rm end}$ be the endpoint transmission subfamily of the complete
experiments in the cited CPRC theorem: it contains arbitrary reference-assisted
inputs and terminal POVMs, has no intermediate logical intervention, and uses the
same complete-record convention as the prescribed task.  Then
\[
\Inst_\tau^{\rm end}\subseteq\Inst_\tau^{\rm CPRC},
\qquad
\Xi_{\rm unc}\big|_{\Inst_\tau^{\rm end}}
=\Xi_\tau\big|_{\Inst_\tau^{\rm end}},
\]
where $\Xi_{\rm unc}$ is the uncoupled reference process in that theorem and
$\Xi_\tau$ is the identity transmission task.  Consequently, the supremum of
the total-variation error over $\Inst_\tau^{\rm end}$ is the endpoint task error
used in Proposition~\ref{prop:replacement-baseline}.
\end{lemma}
\begin{proof}
The endpoint family is obtained from a complete CPRC experiment by removing all
intermediate instruments while retaining the input reference and terminal
readout, so it is a zero-intervention subfamily of the cited experiment class.
With the centered system--environment interaction removed, the reference process
propagates the input and its reference unchanged to the terminal interface.
After the prescribed frame compensation this is exactly the identity channel
with the same terminal record law, which proves the second equality.  The final
statement follows by taking the same supremum over this common task family.
\end{proof}

\begin{corollary}[Budget-constrained CPRC performance]
\label{cor:budget-cprc}
Let $P\in\mathfrak C_{a\to\tau}(B)$ be a CPRC protocol satisfying all hypotheses of the cited complete-experiment theorem, with the prescribed task equal to its uncoupled reference.  Then the budgeted task error obeys
\begin{equation}
 E_{\tau,a,B}(\Lambda)
 \le \eps_\tau(P\Lambda)
 \le \min\left\{1,4\int_0^Tdt\int_0^t ds\,g(t)g(s)b_P(s,t)\right\}.
 \label{eq:budget-cprc}
\end{equation}
For the transmission task and architecture of Proposition~\ref{prop:replacement-baseline}, Lemma~\ref{lem:cprc-transmission-bridge} identifies the task family and reference process on the zero-intervention subfamily.  Thus, if the displayed certificate is $u_P$, it also gives $R_{\tau,a,B}(\Lambda)\ge[1-d^{-2}-u_P]_+$.  Increasing a nominal refresh rate or diffusion strength gives a certificate within a fixed budget only when the physical implementation remains admissible at that budget.
\end{corollary}
\begin{proof}
The first inequality follows from the infimum in Eq.~\eqref{eq:budgeted-error}.  The second is the cited complete-experiment theorem applied to the selected protocol and the fixed task family.  Subtracting this achievable certificate from the exact free baseline gives the resource lower bound.  For the transmission statement, apply Lemma~\ref{lem:cprc-transmission-bridge}; hence $u_P$ bounds the same $\eps_\tau$ that defines $E_{\tau,a,B}$, and subtracting it from $1-d^{-2}$ yields a bound on the resource quantity.  Feasibility of each protocol is determined by its implementation cost, independently of this error estimate.
\end{proof}

For transmission, a valid CPRC certificate $u_P<1-d^{-2}$ certifies
strictly positive capability.  In the ideal holding model,
Eq.~\eqref{eq:refresh-budget-certificate} gives the sufficient condition
\begin{equation}
 N>\frac{2(1-N_{\rm fr}^{-1})g_*^2T^2}{1-d^{-2}},
 \qquad N\le N_{\max},
 \label{eq:positive-cprc-capability}
\end{equation}
provided the remaining hardware caps also admit that schedule and its
stated envelope.  For four Pauli frames and a qubit task this becomes
$N>2g_*^2T^2$; equivalently, the maximum-interval certificate is positive
when $g_*^2Th_*<1/2$.  A finite-transition implementation uses its actual
envelope and adds the certified implementation discrepancy of Proposition~\ref{prop:implementation-stability} before
testing this baseline threshold.

\section{Diagnostics and relation to earlier resource theories}
\label{sec:diagnostics}

The task error is the only optimization target in this theory.  Several quantities remain useful, but they are secondary diagnostics.  The cPT coefficients $C_{\boldsymbol\nu}^{(n)}$ resolve which time-ordered insertion sectors carry the residual error.  A particle-amplitude difference can be bounded sector by sector under the hypotheses of a chosen dilation, but an individual complex amplitude is not a probability and need not be monotone under free operations.  Conditional mutual information can reveal how a control changes correlations across a temporal cut, but its monotonicity depends on the exact process interface and on the class of allowed superprocesses.  Environmental memory and CMI therefore explain where the task error comes from; they do not define independent resources in the present theory.

The contraction underlying our ordering is proved on complete experiments
through physical pullback.  It does not assume data processing for a
distance between process Choi operators, in keeping with the distinction
between Choi and operational distinguishability in Ref.~\cite{zambon2024}.

The superprocess resource theories of Refs.~\cite{berk2021,berk2023} identify valuable temporal structure through their chosen free transformations; their noise-reduction protocols exploit temporal correlations before presenting a controlled process at a coarser interface.  The distinction between physical control followed by interface restriction and passive restriction alone is essential here.  Proposition~\ref{thm:fine-advantage} compares capability for a fixed output task when every coarse protocol is also available at the fine input interface.  It does not forbid a fine control from producing a cleaner endpoint process after coarse-graining.  Indeed the echo in Theorem~\ref{thm:strict-access} produces an ideal endpoint channel from $A$, whereas passive coarse-graining produces $\mathcal D_\phi$.

Thus temporal purification in the present architecture means converting an internally accessible noisy process to a lower-error task process through control and subsequent restriction, with the control budget charged.  It can consume internal access while improving raw endpoint performance.  The budgeted monotone does not increase, because the input capability already optimizes over that same conversion.  This reconciles error reduction by a free control with the resource order, and keeps the complete-record error as the only performance target.  The choice of monotone remains tied to the control class, consistent with Refs.~\cite{berk2025monotones,berk2025quantifiers}.  Randomized compiling and random dynamical decoupling supply established control mechanisms \cite{wallman2016,viola2005}; the complete-record continuous bound used here belongs to Ref.~\cite{wang2026cprc}.

The result also explains why fixed-input, fixed-readout tests are insufficient as a resource definition.  Such a test may assign zero error to an operation that discards the input and prepares the one tested output.  The task interface must therefore include the logical interventions whose behavior is to be protected, and the ideal comparison must use the same complete experiment on both processes.

\section{Discussion and scope}
\label{sec:discussion}

The main conclusion is operational: an uncontrolled endpoint channel cannot determine the task capability accessible through internal time controls. The two-window theorem proves this failure of identification at a finite, explicitly counted control budget and establishes strict fine-versus-coarse advantage for one and the same process. The common endpoint channel fixes neither the fine optimum nor the resource advantage. Its best possible capability prediction still has worst-case error at least $p/2$. The fixed-total-phase family resolves the mechanism, and the sharp action threshold quantifies the investment needed for a nontrivial perfect echo.

The resource theory turns these physical statements into a consistent conversion order. Complete entrance erasure supplies a free family at every compatible access layer. Causal slot substitution and physical experiment pullback give identity and composition closure, free-set invariance, and budgeted monotonicity. For reference-assisted transmission, the benchmark $1-d^{-2}$ calibrates a process's ability to retain the unknown input. The closed priced task subtheory applies to realizable cPTs at feasible budgets in the declared access and control architecture. For general tasks, the same construction orders optimized task error $E$; the calibrated advantage $R$ is monotone across interfaces when their free baselines agree. Thus the resource ordering is specified by the task, access structure, and available control hardware.

Task covariance establishes the connection to logical gates. Same-dimensional unitary targets have the exact free benchmark $1-d^{-2}$ whenever the budget admits a depolarizing free realization. Output-frame covariance and joint input--output transport identify the distinct interface and cost correspondences needed to transfer optimized errors. The concrete two-qubit model gives a direct physical result: a finite echo pair and a paid boundary gate realize the same $0/p/p$ separation for every $U\in U(2)$ whenever the operations fit the common budget. For $U=-iX$, one pulse implements both the gate and the echo at action $\pi/2$, illustrating why an additive gate charge is not automatic. This construction supplies a gate-level access advantage independently of an abstract exact-budget translation.

Implementation stability makes the comparison quantitative beyond the nominal control network. Complete-record discrepancy bounds apply to any fixed task, with coherent and dissipative estimates evaluated on the full retained dynamics. They include pulse, controller, instrument, and record imperfections while keeping the environmental action assigned to the cPT explicit. A common commuting window interaction preserves exact endpoint equality and a strictly positive gap below the stated threshold. A covariant control-layer depolarizer gives the exact attenuated gap $\alpha p$. These are distinct physical comparisons, each with its declared supports, interfaces, and control costs.

Continuous process language carries the implementation information needed for these comparisons: finite instrument support, internal control ports, generator peak and action, full record laws, and compatible limits of control paths. The exact separation is realized on a finite-dimensional system--environment network. The continuous pullback results additionally treat open instruments and their specified completions, preserving the causal, ancillary, and event structure. This joins the analytic model to control protocols whose physical specification contains more information than an endpoint channel.

CPRC provides such protocols when its physical instrument pullback obeys the declared architecture. Operational contraction and compatible completion give a conditional extension of these conversions. The complete-record dynamical theorem of Ref.~\cite{wang2026cprc}, under centered bounded interaction, logical initial independence, autonomous control, and covariant instruments, supplies achievable-error certificates. The theory converts those certificates into budgeted capability guarantees and specifies how finite implementation discrepancies enter them. Exact optima belong to the finite-window control class; mixing and refresh bounds certify the selected CPRC implementations.

The next computational task is to obtain matching upper and lower certificates for larger declared control classes. High-dimensional searches can use realizable protocols for upper bounds and the fixed task interface for lower bounds. Families of tasks and supplies of quantum memory or entanglement can likewise be introduced through typed interfaces and costs, with their own compatible free benchmarks. CMI, environmental memory, and particle amplitudes can guide such searches by locating the temporal mechanisms that contribute to complete-record error.

The results identify a measurable distinction between endpoint noise and controllable temporal structure. Temporal access determines which interactions can be canceled; the control budget determines whether the required conversion is feasible. Protecting a logical task therefore requires characterization of the internal process and its admissible control architecture, together with an error criterion that compares the same complete experiment with its prescribed ideal.

\renewcommand{\href}[2]{#2}
\def\Eprint#1#2{#2}
\bibliography{references}

\end{document}